\documentclass[11pt]{article}
\usepackage{rheaj}
\usepackage{setspace}
\usepackage{xparse}
\newcommand{\mypara}[1]{\medskip \noindent {\bf #1}}

\newcommand{\racke}{R\"{a}cke\xspace}

\newcommand{\load}{\textnormal{load}}
\newcommand{\rload}{\textnormal{rload}}
\newcommand{\maxflow}{\textnormal{maxflow}}

\begin{document}
\title{Approximation Algorithms for Network Design \\ in Non-Uniform
  Fault Models\footnote{A preliminary version of this paper appeared
    in Proc.\ of ICALP, 2023 \cite{ChekuriJ23}. This version contains
    extensions and improved approximation ratios for some problems via
    an alternate approach.}
}
\author{Chandra Chekuri \thanks{Dept. of Computer Science, Univ. of
    Illinois, Urbana-Champaign, Urbana, IL 61801. {\tt
      chekuri@illinois.edu}. Supported in part by NSF grants CCF-1910149 and CCF-1907937.}
\and
Rhea Jain \thanks{Dept. of Computer Science, Univ. of Illinois,
  Urbana-Champaign, Urbana, IL 61801. {\tt rheaj3@illinois.edu}. Supported in part by NSF grant CCF-1910149.}
}
\date{\today}
\maketitle

\begin{abstract}
  The Survivable Network Design problem (SNDP) is a well-studied
  problem, (partly) motivated by the design of networks that are
  robust to faults under the assumption that any subset of edges up to
  a specific number can fail.  We consider \emph{non-uniform} fault
  models where the subset of edges that fail can be specified in
  different ways. Our primary interest is in the flexible graph
  connectivity model
  \cite{Adjiashvili13,AdjiashviliHM22,BoydCHI22}, in
  which the edge set is partitioned into \emph{safe} and \emph{unsafe}
  edges. The goal is to design a network that has desired connectivity
  properties under the assumption that only unsafe edges up to a
  specific number can fail.
  We also discuss the bulk-robust model
  \cite{AdjiashviliSZ15,Adjiashvili15} and the relative survivable
  network design model \cite{DinitzKK22,DinitzKKN23}.
  While SNDP admits a $2$-approximation \cite{Jain01}, the
  approximability of problems in these more complex models is much
  less understood even in special cases. We make two contributions.

  Our first set of results are in the flexible graph connectivity
  model. Motivated by a conjecture that a constant factor
  approximation is feasible when the robustness parameters are
  fixed constants, we consider two important special cases, namely the
  single pair case, and the global connectivity case. For both these,
  we obtain constant factor approximations in several parameter
  ranges of interest. These
  are based on an augmentation framework and via decomposing the families
  of cuts that need to be covered into a small number of uncrossable
  families.
  
  Our second set of results are poly-logarithmic approximations for
  the bulk-robust model \cite{AdjiashviliSZ15} when the ``width'' of
  the given instance (the maximum number of edges that can fail in any
  particular scenario) is fixed. Via this, we derive corresponding
  approximations for the flexible graph connectivity model and the
  relative survivable network design model.  The results are obtained
  via two algorithmic approaches and they have different tradeoffs in terms of
  the approximation ratio and generality.
\end{abstract}

\section{Introduction}
\label{sec:intro}
The Survivable Network Design Problem (SNDP) is an important problem
in combinatorial optimization that generalizes many well-known
problems related to connectivity and is also motivated by practical
problems related to the design of fault-tolerant networks. The input
to this problem is an undirected graph $G=(V,E)$ with non-negative
edge costs $c: E \to \R_+$ and a collection of source-sink pairs
$(s_1,t_1),\ldots,(s_h,t_h)$, each with an integer connectivity requirement $r_i$.
The goal is to find a minimum-cost subgraph $H$ of $G$ such
that $H$ has $r_i$ connectivity for each pair $(s_i,t_i)$.
We focus on edge-connectivity requirements in this paper.\footnote{In the literature
the term EC-SNDP and VC-SNDP are used to distinguish edge and vertex connectivity
requirements. In this paper we use SNDP in place of EC-SNDP.}
SNDP contains as special cases classical problems such as $s$-$t$ shortest path,
minimum spanning tree (MST), minimum $k$-edge-connected subgraph
($k$-ECSS), Steiner tree, Steiner forest and several others. It is
NP-Hard and APX-Hard to approximate. There is a $2$-approximation
via the iterated rounding technique \cite{Jain01}.

A pair $(s,t)$ that is $k$-edge-connected in $G$ is robust to the
failure of \emph{any} set of $k-1$ edges. In various settings,
the set of edges that can fail can be correlated and/or exhibit
non-uniform aspects. We are interested in network design
in such settings, and discuss a few models of interest that
have been studied in the (recent) past. We start with the flexible
graph connectivity model (flex-connectivity for short) that was the
impetus for this work.

\mypara{Flexible graph connectivity:} In this model, first
introduced by Adjiashvili \cite{Adjiashvili13} and studied in several
recent papers
\cite{AdjiashviliHM22,AdjiashviliHMS20,BoydCHI22,BansalCGI22,Bansal23},
the input is an edge-weighted undirected graph $G=(V,E)$ where the
edge set $E$ is partitioned to \emph{safe} edges $\calS$ and
\emph{unsafe} edges $\calU$. The assumption, as the names suggest, is
that unsafe edges can fail while safe edges cannot. We say that a
vertex-pair $(s,t)$ is $(p,q)$-flex-connected in a subgraph $H$ of $G$
if $s$ and $t$ are $p$-edge-connected after deleting from $H$ any
subset of at most $q$ unsafe edges. The input, as in SNDP, consists of $G$ and
$h$ source-sink pairs; the $i$'th pair now specifies a $(p_i,q_i)$-flex-connectivity
requirement. The goal is to find a min-cost subgraph $H$ of $G$
such that for each $i \in [h]$, $s_i$ and $t_i$ are $(p_i,q_i)$-flex-connected
in $H$.  We refer to this as the Flex-SNDP problem. Note that
Flex-SNDP generalizes SNDP in two ways. If all edges are safe
($E = \calS$), then $(p,0)$-flex-connectivity is equivalent to
$p$-edge-connectivity. Similarly, if all edges are unsafe ($E =
\calU$), then $(1,q-1)$-flex-connectivity is equivalent to
$q$-edge-connectivity.

\mypara{Bulk-robust network design:} This fairly general
non-uniform model was introduced by Adjiashvili, Stiller and Zenklusen
\cite{AdjiashviliSZ15}.  Here an explicit \emph{scenario} set $\Omega
= \{F_1,F_2,\ldots,F_m\}$ is given as part of the input where each
$F_j \subseteq E$. The goal is to find a min-cost subgraph $H$ of $G$
such that each of the given pairs $(s_i,t_i)$ remains connected in $H
- F_j$ for each $j \in [m]$. We consider a slight generalization of
this problem in which each scenario is now a pair $(F_j,
\mathcal{K}_j)$ where $\mathcal{K}_j$ is a set of source-sink
pairs. As earlier, the goal is to find a min-cost subgraph $H$ of $G$
such that for each $j \in [m]$, each pair $(s_i,t_i)$ in $\mathcal{K}_j$ is
connected in $H-F_j$. The \emph{width} of the failure scenarios is
$\max_{1\le j \le \ell} |F_j|$. We use Bulk-SNDP to refer to this
problem.

The advantage of the bulk-robust model is that one can specify
arbitrarily correlated failure patterns, allowing it to capture many
well studied problems in network design. We observe that SNDP and
Flex-SNDP problem can be cast as special cases of Bulk-SNDP model
where the width is $\max_i (r_i-1)$ in the former case and $\max_i
(p_i+q_i-1)$ in the latter case. The slight generalization on
Bulk-SNDP described above also allows us to model a new problem
recently proposed by Dinitz, Koranteng, and Kortsarz
\cite{DinitzKK22} called \textbf{Relative Survivable Network Design}
(RSNDP). This problem allows one to ask for higher connectivity even
when the underlying graph $G$ has small cuts. The input is an
edge-weighted graph $G=(V,E)$ and source-sink pairs $(s_i,t_i)$ each
with requirement $r_i$; the goal is to find a min-cost subgraph $H$ of
$G$ such that for each $F \subseteq E$ with $|F| < r_i$,
$(s_i,t_i)$ is connected in $H-F$ if $s_i$ and $t_i$ are connected in
$G-F$. It is easy to see that RSNDP is a special case of Bulk-SNDP
with width at most $\max_i (r_i-1)$. A disadvantage of Bulk-SNDP is that
scenarios have to be explicitly listed, while the other models
discussed specify failure scenarios implicitly. However, when
connectivity requirements are small/constant, one can reduce to
Bulk-SNDP by explicitly listing the failure sets.

We are also interested in the generalization of non-uniform models to 
Group Connectivity (also known as Set or Generalized Connectivity),
first introduced by \cite{ReichW89} and later studied in approximation by 
\cite{GargKR98} and several others; Section \ref{subsec:related_work}
provides more details.
In this setting, instead of source-sink pairs 
of vertices, we are given source-sink pairs of \emph{sets} $S_i, T_i \subseteq V$. 
We say $S_i$ and $T_i$ are $k$-edge-connected if there is a path from 
$S_i$ to $T_i$ even after the removal of any set of $k-1$ 
edges. The definitions of flex-connectivity, bulk-robust network design,
and relative survivable network design can all be extended to the group 
setting by requiring edge connectivity between set pairs instead of vertex pairs.

While SNDP admits a $2$-approximation, the approximability of network
design in the preceding models is not well-understood.  The known
results mostly focus on two special cases: (i) the single pair case
where there is only one pair $(s,t)$ with a connectivity requirement
and (ii) the spanning or global connectivity case when all pairs of
vertices have identical connectivity requirement.  Even in the single
pair case, there are results that show that problems in the
non-uniform models are hard to approximate to poly-logarithmic or
almost-polynomial factors when the connectivity requirement is not
bounded \cite{AdjiashviliSZ15,AdjiashviliHMS20}. Further, natural LP
relaxations in some cases can also be shown to have large integrality
gaps (see Appendix \ref{sec:integrality_gap_1k}). 
Motivated by these negative results and practical
considerations, we focus our attention on Flex-SNDP when the max
connectivity requirement $p,q$ are small, and similarly on Bulk-SNDP
when the width is small. Other network design problems with similar
hardness results have admitted approximation ratios that depend on the
max connectivity requirement, for example, VC-SNDP
\cite{ChakrabortyCK08,ChuzhoyK12,Nutov12} and the single-pair case of
Bulk-Robust \cite{AdjiashviliSZ15}.

\subsection{Our Contributions and Comparison to Existing Work}
\label{subsec:contribution}
We are mainly motivated by Flex-SNDP and insights for it via Bulk-SNDP.
We make two broad contributions. Our first set of results is on
special cases of Flex-SNDP for which we obtain constant factor
approximations. Our second contribution is a poly-logarithmic approximation
for Flex-SNDP, Bulk-SNDP, and RSNDP when the requirements are small. 
Some of these results were initially presented in a conference version 
\cite{ChekuriJ23}; we extend and improve on those results here.

We use the terminology $(p,q)$-Flex-ST to refer to the single-pair problem
with requirement $(p,q)$. We use the term $(p,q)$-FGC to refer to the
spanning/global-connectivity problem where all pairs of vertices
have the $(p,q)$-flex-connectivity requirement (the term FGC is to be
consistent with previous usage \cite{AdjiashviliHM22,BoydCHI22}).

\mypara{$(p,q)$-FGC:} Adjiashvili et al.\ \cite{AdjiashviliHM22}
considered $(1,1)$-FGC and obtained a constant factor approximation
that was subsequently improved to $2$ by Boyd et
al. \cite{BoydCHI22}. \cite{BoydCHI22} obtained several results for
$(p,q)$-FGC including a $4$-approximation for $(p,1)$-FGC, a
$(q+1)$-approximation for $(1,q)$-FGC, and an $O(q \log
n)$-approximation for $(p,q)$-FGC. The first non-trivial case of small
$p,q$ for which we did not know a constant factor is $(2,2)$-FGC. We
prove several results that, as a corollary, yield constant factor
approximation for small values of $p,q$.

\begin{theorem}
\label{thm:introfgc}
  For any $q \ge 0$ there is a $(2q+2)$-approximation for $(2,q)$-FGC.
  For any $p \ge 1$  there is a $(2p+4)$-approximation for $(p,2)$-FGC,  and a
  $(4p+4)$-approximation for $(p,3)$-FGC.  Moreover, for
  all \emph{even} $p \ge 2$ there is a $(6p+4)$-approximation for $(p,4)$-FGC.
\end{theorem}

\begin{remark}
  In independent work Bansal et al.\ \cite{BansalCGI22} obtained an
  $O(1)$-approximation for $(p,2)$-FGC for all $p \ge 1$ ($6$ when $p$
  is even and $20$ when $p$ is odd).
  In more recent work Bansal \cite{Bansal23} obtained an $O(1)$
  approximation for $(p,3)$-FGC for all $p \ge 1$.
\end{remark}

\mypara{$(p,q)$-Flex-ST:} Adjiashvili et al.\ \cite{AdjiashviliHM22}
considered $(1,q)$-Flex-ST and $(p,1)$-Flex-ST and obtained several
results. They described a $q$-approximation for $(1,q)$-Flex-ST and a
$(p+1)$-approximation for $(p,1)$-Flex-ST; when $p$ is a fixed
constant they obtain a $2$-approximation.  Also implicit in
\cite{AdjiashviliSZ15} is an $O(q(p+q)\log n)$-approximation algorithm
for $(p,q)$-Flex-ST that runs in $n^{O(p+q)}$-time.  No constant
factor approximation was known when $p,q \ge 2$ with $(2,2)$-Flex-ST
being the first non-trivial case. We prove a constant factor approximation
for this and several more general settings via the following theorem.

\begin{theorem}
\label{thm:intro-flex-st}
  For all $p,q$ where $(p+q) > pq/2$, there is an $O((p+q)^{O(p)})$-approximation 
  algorithm for $(p,q)$-Flex-ST that runs in $n^{O(p+q)}$ time.
  In particular, there is an $O(1)$ approximation for $(p,2)$ and $(2,q)$-Flex-ST 
  when $p,q$ are fixed constants.
\end{theorem}

\mypara{Flex-SNDP, Bulk-SNDP, and RSNDP:} We show that these problems
admit poly-logarithmic approximation algorithms when the widths/connectivity
requirements are small. No previous approximation 
algorithms were known for SNDP versions of
flexible graph connectivity (with both $p,q \ge 2$) or bulk-robustness, and 
previous approximations for RSNDP were limited to $k \leq 3$ \cite{DinitzKKN23}. 
We rely on two different algorithmic approaches to obtain these results. The first
uses a cost-sharing argument to show a reduction to the Hitting Set
problem. \footnote{This approach, Theorem \ref{thm:bulk-sndp} and corollaries 
\ref{cor:flex-sndp},\ref{cor:rsndp} are new to this version, and
provide improved approximation rations when compared to the ones in
\cite{ChekuriJ23}.}

\begin{theorem}
\label{thm:bulk-sndp}
  There is a randomized algorithm that yields an $\tilde O(k^2 \log^2 
  n)$-approximation for Bulk-SNDP on instances with width at most $k$ and runs in 
  $n^{O(k)}$ time.
\end{theorem}

\begin{corollary}
\label{cor:flex-sndp}
  There is a randomized algorithm that yields an $\tilde O(q(p+q)\log^2
  n)$-approximation for Flex-SNDP when $(p_i, q_i) \leq (p,q)$ for all 
  pairs $(s_i, t_i)$ and runs in $n^{O(p+q)}$-time.
\end{corollary}

\begin{corollary}
\label{cor:rsndp}
  There is a randomized algorithm that yields an $\tilde O(k^2 \log^2
  n)$-approximation for RSNDP where $k$ is the maximum connectivity requirement
  and runs in $n^{O(k)}$ time.
\end{corollary}

The second set of algorithms are based on natural LP relaxations 
(see Section \ref{sec:prelim}). Although the approximation ratios in these 
algorithms are slightly worse than those above, they have some additional
benefits.
First, the bottleneck in the running time is based on the time it takes to 
solve the LP relaxations of the given problem, 
giving improved runtimes (as compared to Theorem \ref{thm:bulk-sndp}
and Corollaries \ref{cor:flex-sndp},\ref{cor:rsndp}) in some settings. 
Second, all approximation ratios are with respect to the optimal \emph{fractional}
solution, providing upper bounds on the integrality gaps of the LP relaxations.
Third, these results extend to the group connectivity generalizations of 
Bulk-SNDP, Flex-SNDP, and RSNDP. For these problems, we let $r$ denote 
the total number of terminal pairs. Note that unlike the SNDP variants,
$r$ is not necessarily polynomial in $n$.
\footnote{Theorem \ref{thm:bulk-group} and corollaries 
\ref{cor:flex-group},\ref{cor:rgroup} were discussed in \cite{ChekuriJ23} for 
the SNDP problems. In this paper we explicitly show the extension to 
group connectivity.}

\begin{theorem}
\label{thm:bulk-group}
  There is a randomized algorithm that yields an 
  $O((\log r + k\log n)k^3\log^6 n)$- approximation for Group
  Bulk-SNDP on instances with width at most $k$ and runs in 
  expected polynomial time.
\end{theorem}

\begin{corollary}
\label{cor:flex-group}
  There is a randomized algorithm that yields an 
  $O((\log r + (p+q)\log n)(p+q)^3\log^6 n)$- approximation for Group
  Flex-SNDP when $(p_i,q_i) \le (p,q)$ for all pairs 
  $(s_i,t_i)$ and runs in expected $n^{O(q)}$-time.
\end{corollary}

\begin{corollary}
\label{cor:rgroup}
  There is a randomized algorithm that yields an 
  $O((\log r + k\log n)k^3\log^6 n)$- approximation for Group
  RSNDP where $k$ is the maximum connectivity requirement and runs in 
  expected polynomial time.
\end{corollary}

\subsection{Related work}
\label{subsec:related_work}
Network design has substantial literature. We describe closely related  work
and results to put ours in context.

\mypara{SNDP and related connectivity problems:} SNDP is a
canonical problem in network design for connectivity that captures
many problems. We refer the reader to some older
surveys \cite{GuptaK11,KortsarzN10} on approximation algorithms for connectivity problems,
and several recent papers with exciting progress on TSP and weighted Tree and Cactus augmentation.
Frank's books is an excellent source for polynomial-time solvable exact algorithms \cite{Frank-book}. For
SNDP, the augmentation approach was pioneered in \cite{WilliamsonGMV95},
and was refined in \cite{GoemansGPSTW94}. These led to $2H_k$
approximation where $k$ is the maximum connectivity
requirement. Jain's iterated rounding approach \cite{Jain01} obtained
a $2$-approximation. See \cite{FleischerJW06,CheriyanVV06} for
element-connectivity  and \cite{ChuzhoyK12,Nutov12} for 
vertex-connectivity requirements. 

\mypara{Flexible Graph Connectivity:} Flexible graph connectivity
has been a topic of recent interest, although the model was introduced
earlier in the context of a single pair \cite{Adjiashvili13}.
Adjiashvili, Hommelsheim and M\"uhlenthaler \cite{AdjiashviliHM22}
introduced FGC (which is the same as $(1,1)$-FGC) and pointed out that it
generalizes the well-known MST and 2-ECSS problems.
Several approximation algorithms for various special cases of FGC 
and Flex-ST were obtained by Adjiashvili et al. \cite{AdjiashviliHM22}
and Boyd et al. \cite{BoydCHI22}, as described in Section \ref{subsec:contribution}.

Adjiashvili et al.\ \cite{AdjiashviliHMS20} also showed hardness results
in the single pair setting. 
They prove that $(1,k)$-Flex-ST in
directed graphs is at least as hard as directed Steiner tree which
implies poly-logarithmic factor inapproximability
\cite{HalperinK03}. They prove that $(k,1)$-Flex-ST in directed graphs is
at least as hard to approximate as directed Steiner forest (which has
almost polynomial factor hardness \cite{DodisK99}).  The hardness
results are when $k$ is part of the input and large, and show that
approximability of network design in this model is substantially
different from the edge-connectivity model.

\mypara{Bulk-Robust Network Design:} This model was initiated in
\cite{AdjiashviliSZ15} where they obtained an $O(\log n + \log m)$
approximation for the Bulk-Robust spanning tree problem.  The same paper
shows that the \emph{directed} single pair problem (Bulk-Robust
shortest path) is very hard to approximate. The hardness reduction
motivated the definition of width and the paper describes an
$O(k^2 \log n)$-approximation for Bulk-Robust shortest path via a 
reduction to the Hitting Set problem, and the use of the augmentation
approach, that we build upon here. For the special case of $k = 2$
the authors obtain an $O(1)$-approximation. Adjiashvili
\cite{Adjiashvili15} showed that if the graph is planar then one can
obtain an $O(k^2)$-approximation for both Bulk-Robust shortest path
and spanning tree problems --- he uses the augmentation approach from
\cite{AdjiashviliSZ15} and shows that the corresponding covering
problem in each augmentation phase corresponds to a Set Cover problem
that admits a constant factor approximation. As far as we are aware,
there has not been any progress beyond these special cases.

\mypara{Relative Network Design;} This model was introduced in 
recent work \cite{DinitzKK22}. The authors obtain a $2$-approximation
for the spanning case via the iterated rounding technique
even though the requirement function is not skew-supermodular.
They also obtain a simple $2$-approximation when the maximum requirement is $2$.
They obtain a $\frac{27}{4}$-approximation for the $(s,t)$-case when the maximum
demand is $3$. In more recent work \cite{DinitzKKN23}, the authors obtain 
a 2-approximation for general RSNDP when the maximum demand is 3, and a
$2^{O(k^2)}$-approximation for the $(s,t)$-case when the maximum demand is $k$.

\mypara{Group Connectivity:}
Group connectivity was first studied in the single-source, one-connectivity 
setting known as Group Steiner Tree. Garg et al.\ \cite{GargKR98}
described a randomized algorithm to round a fractional solution to a
cut-based LP relaxation when $G$ is a tree --- it achieves an $O(\log n
\log k)$-approximation. This result is essentially tight 
\cite{HalperinK03,HalperinKKSW07}.
Using tree embeddings with $O(\log n)$ expected distortion 
\cite{Bartal98,FRT03}, one can extend the algorithm to general graphs
and obtain an $O(\log^2 n \log k)$-approximation. 
For general graphs an $O(\log^2 k)$ approximation can be obtained in
quasi-polynomial time \cite{ChekuriP05,GrandoniLL19,GhugeN22}.
Set Connectivity is a generalization of group Steiner tree problem to 
multiple source-sink pairs of sets. This was first introduced by Alon et 
al. in the online setting \cite{aaabn06}. In the offline setting, Chekuri et al. 
\cite{ChekuriEGS11} obtained a polylogarithmic approximation ratio with 
respect to the LP by extending ideas from Group Steiner Tree. 

The higher connectivity setting, known as Survivable Set Connectivity,
was studied by \\Chalermsook, Grandoni and Laekhanukit \cite{CGL15},
motivated by earlier work in \cite{GuptaKR10}. Here each pair
$(S_i,T_i)$ has a connectivity requirement $r_i$ which implies that
one seeks $r_i$ edge-disjoint paths between $S_i$ and $T_i$ in the
chosen subgraph $H$. \cite{CGL15} obtained a bicriteria-approximation
via \racke tree and group Steiner tree rounding.  The recent work of
Chen et al. \cite{ChenLLZ22} uses related but more sophisticated ideas
to obtain the first true approximation for this problem. They refer to
the problem as Group Connectivity problem and obtain an
$O(r^2 \log r \log^6 n(r\log n + \log q))$-approximation where 
$r = \max_i r_i$ is the maximum connectivity requirement and $q$ is the 
number of groups
(see \cite{ChenLLZ22} for more
precise bounds). 

\subsection{Overview of techniques} As we remarked, the non-uniform models have been difficult to handle
for existing algorithmic techniques. The structures that underpin the
known algorithms for SNDP (primal-dual \cite{WilliamsonGMV95} and
iterated rounding \cite{Jain01}) are
skew-supermodularity of the requirement function and
submodularity of the cut function in graphs. Since non-uniform models
do not have such clean structural properties, these known techniques
cannot be applied directly. Another technique for network design,
based on several previous works, is augmentation. In the augmentation
approach we start with an initial set of edges $F_0$ that partially
satisfy the connectivity constraints. We then augment $F_0$ with a set
$F$ in the graph $G-F_0$; the augmentation is typically done to
increase the connectivity by one unit for pairs that are not yet
satisfied. We repeat this process in several stages until all
connectivity requirements are met. The utility of the 
augmentation approach is that it allows one to reduce a
higher-connectivity problem to a series of problems that solve a
potentially simpler $\{0,1\}$-connectivity problem.  An important tool
in this area is a $2$-approximation for covering an \emph{uncrossable}
function (a formal definition is given in
Section~\ref{sec:fgc}) \cite{WilliamsonGMV95}. 

In trying to use the
augmentation approach for Flex-SNDP and its special cases, we see that
the resulting functions are usually not uncrossable. To prove 
Theorems~\ref{thm:introfgc} and \ref{thm:intro-flex-st}, we
overcome this difficulty by decomposing the family of cuts to be
covered in the augmentation problem into a sequence of cleverly chosen
uncrossable subfamilies. Our structural results hold for certain range
of values of $p$ and $q$ and hint at additional structure that may be
available to exploit in future work. Boyd et al. also show a
connection to \emph{capacitated} network design (also implicitly in \cite{AdjiashviliHMS20})
which has been studied in several works
\cite{GoemansGPSTW94,CarrFLP00,ChakCKK15,ChakKLN13}. This model
generalizes standard edge connectivity by allowing each edge $e$ to
have an integer capacity $u_e \geq 1$. One can reduce capacitated
network design to standard edge connectivity by replacing each edge
$e$ with $u_e$ parallel edges, blowing up the approximation factor by
$\max_e u_e$. Boyd et al. show that $(1,k)$-Flex-SNDP and
$(k,1)$-Flex-SNDP can be reduced to Cap-SNDP with maximum capacity
$k$. While this reduction does not extend when $p,q \geq 2$, it
provides a useful starting point that we exploit for Theorem~\ref{thm:intro-flex-st}.

For the more general SNDP problems, we continue to employ the augmentation 
framework, but use completely different algorithmic approaches to solve these 
problems. To prove Theorem~\ref{thm:bulk-sndp}, we use a characterization of the 
\emph{edge-cuts} of the existing graph given by Gupta, Ravi, and Ravishankar 
\cite{gupta2009online}. They show that the augmentation problem in the online 
version of the SNDP problem can be reduced to a polynomial-sized Hitting Set instance. 
We show that these ideas extend to the general Bulk-SNDP setting as well. 
To prove Theorem~\ref{thm:bulk-group},
we rely on a recent novel
framework of Chen, Laekhanukit, Liao, and Zhang \cite{ChenLLZ22} to
tackle survivable network design in group connectivity setting. They
used the seminal work of \racke \cite{Racke08} on probabilistic
approximation of capacitated graphs via trees, and the group Steiner
tree rounding techniques of Garg, Konjevod and Ravi \cite{GargKR98},
and subsequent developments~\cite{GrandoniLL19}. We adapt their ideas 
to handle the augmentation problem for Bulk-SNDP.  

\mypara{Organization:} Section~\ref{sec:prelim} sets up the
relevant background on the LP relaxations for Flex-SNDP and
Bulk-SNDP and on the augmentation framework. 
Section~\ref{sec:fgc} provides the proofs of Theorems~\ref{thm:introfgc}
and \ref{thm:intro-flex-st}. Section~\ref{sec:sndp_cutcover} provides the proofs of
Theorems~\ref{thm:bulk-sndp} and the resulting corollaries
\ref{cor:flex-sndp} and \ref{cor:rsndp}. Section \ref{sec:sndp-racke} provides the 
proofs of Theorems \ref{thm:bulk-group} and the resulting corollaries 
\ref{cor:flex-group} and \ref{cor:rgroup}.

\section{Preliminaries}
\label{sec:prelim}
Throughout the paper we will assume that we are given an undirected graph
$G=(V,E)$ along with a cost function $c: E \to \R_{\geq 0}$. When we say that 
$H$ is a subgraph of $G=(V,E)$ we
implicitly assume that $H$ is an edge-induced subgraph, i.e. $H=(V,F)$ for some 
$F \subseteq E$. For any
subset of edges $F \subseteq E$ and any set $S \subseteq V$
we use the notation $\delta_F(S)$ to denote the set of edges in $F$ that
have exactly one endpoint in $S$.  We may drop $F$ if it is clear from the
context. For all discussion of flex-connectivity, we will let $\calS$ denote the 
set of safe edges and $\calU$ denote the set of unsafe edges. 

\subsection{LP Relaxations}
\mypara{Flex-SNDP:} We describe an LP relaxation for
$(p,q)$-Flex-SNDP problem. 
We assume for simplicity that $(p_i, q_i)$ is either $(p,q)$ or $(0,0)$ for all 
pairs of vertices. Thus we can describe the input as 
a set of $h$ terminal pairs 
$(s_i, t_i) \subseteq V \times V$ and the goal is to
choose a min-cost subset of the edges $F$ such that in the subgraph
$H=(V,F)$, $s_i$ and $t_i$ are $(p,q)$-flex-connected for any $i \in [h]$. Let 
$\calC = \{S \subset V \mid \exists i \in [h] \text{ s.t. }|S \cap \{s_i, t_i\}| = 1\}$ 
be the set of all vertex sets that separate some terminal pair. 
For a set of edges $F$ to be feasible for the given $(p,q)$-Flex-SNDP instance, 
we require that for all $S \in \calC$, $|\delta_F(S) \setminus B| \ge p$ for any 
$B \subseteq \calU$ with $|B| \le q$. 

We can write cut covering constraints 
expressing this condition, but these constraints are not adequate by themselves.
To see this, consider a cut $S$ with $N$ unsafe edges, each with capacity 
$p/(N-q)$. For all subsets $B$ of $q$ unsafe edges, 
$\sum_{e \in \delta(S) - B} x_e \geq p$. However, the total 
sum of capacities on this cut $\sum_{e \in \delta(S)} x_e = p + \frac{qp}{N-q}$, 
which we can set to be arbitrarily close to $p$ by increasing $N$. Since 
all edges on $S$ are unsafe, any integral solution would need at least $p+q$ 
edges, so the integrality gap is at least $(p+q)/p$.

To improve this LP, we consider the connection to capacitated network design: 
we give each safe edge a capacity of $p+q$, each unsafe edge a capacity of $p$, 
and require $p(p+q)$ connectivity for the terminal pairs; it is not difficult to 
verify that this is a valid constraint. These two sets of constraints yield the 
following LP relaxation with variables $x_e \in [0,1]$, $e \in E$.

\begin{align*}
    \min \sum_{e \in E} c(e)x_e &\\
    \text{subject to} \sum_{e \in \delta(S) - B} x_e &\geq p  
    &S \in \calC, B \subseteq \calU, |B| \le q \\
    (p+q)\sum_{e \in \delta_{\calS}(S)} x_e + p \sum_{e \in \delta_{\calU}(S)}
  x_e &\geq p(p+q) & S \in \calC \\
    x_e & \in [0,1] &e \in E
\end{align*}

The following lemma borrows ideas from \cite{BoydCHI22,ChakCKK15}. 
\begin{lemma}
\label{lemma:lp_polytime_fgc}
  The Flex-SNDP LP relaxation can be solved in $n^{O(q)}$ time. 
  For $(p,q)$-FGC, it can be solved in polynomial time.
\end{lemma}
\begin{proof}
  We show an $n^{O(q)}$ separation oracle for the given
  LP. Suppose we are given some vector $x \in [0,1]^{|E|}$. We first
  check if the capacitated min-cut constraints are satisfied. 
  This can be done in polynomial time by giving every
  safe edge a weight of $p+q$ and every unsafe edge a weight of $p$,
  and checking that the min-cut value is at least $p(p+q)$. If it is
  not, we can find the minimum cut and output the corresponding
  violated constraint. Suppose all capacitated constraints are
  satisfied. Then, for each $B \subseteq \calU, |B| \le q$
  we remove $B$ and check that for each $s,t \in T$, the $s$-$t$ min-cut
  value in the graph $G-B$ with edge-capacities given by $x$ is at least
  $p$. Since there are at most $n^{O(q)}$ such possible sets $B$, 
  we get our desired separation oracle.
  
  In the FGC case, if there is a remaining unsatisfied constraint, then there 
  must be some $S \subset V$ and some
  $B \subseteq \calU$, $|B| \le q$, such that
  $\sum_{e \in \delta(S) - B} x_e < p$. In particular,
  $\sum_{e \in \calS \cap \delta(S) - B} x_e < p$ and
  $\sum_{e \in \calU \cap \delta(S) - B} x_e < p$. We claim that the
  total weight (according to weights $(p+q)$ for safe edges and $p$
  for unsafe edges) going across $\delta(S)$ is at most $2p(p+q)$: at
  most $(p+q)p$ from $\calS \cap (\delta(S) - B)$, at most $p^2$ from
  $\calU \cap (\delta(S) - B)$, and at most $pq$ from $B$. Recall that
  the min-cut of the graph according the weights has already been
  verified to be at least $p(p+q)$. Hence, any violated cut from the
  first set of constraints corresponds to $2$-approximate min-cut.  It
  is known via Karger's theorem that there are at most $O(n^4)$
  $2$-approximate min-cuts in a graph, and moreover they can also be
  enumerated in polynomial time \cite{karger1993global,Karger00}. We
  can enumerate all $2$-approximate min-cuts and check each of them to
  see if they are violated. To verify whether a candidate cut $S$ is
  violated we consider the unsafe edges in $\delta(S) \cap \calU$ and
  sort them in decreasing order of $x_e$ value. Let $B'$ be a prefix
  of this sorted order of size $\min\{q, |\delta(S) \cap \calU|\}$. It
  is easy to see that $\delta(S)$ is violated iff it is violated
  when $B = B'$. Thus, we can verify all candidate cuts efficiently.
\end{proof}

\mypara{Bulk-SNDP:} We can similarly define an LP relaxation for the 
Bulk-SNDP problem. As above, we have a variable $x_e \in [0,1]$ for each edge 
$e \in E$. 
\begin{align*}
    \min \sum_{e \in E} c(e)x_e& &  \\
    \text{subject to } \sum_{e \in \delta(S) \setminus F_j} x_e 
    &\geq 1 &\forall (F_j, \calK_j) \in \Omega, 
    S \text{ separates a terminal pair in } \calK_j \\
    x_e \in [0,1] & 
\end{align*}

\begin{lemma}
  The Bulk-SNDP LP relaxation can be solved in time polynomial in $n$ and $m$.
\end{lemma}
\begin{proof}
  We show a polynomial time separation oracle for the LP. For each scenario 
  $(\calF_j, \calK_j)$, let $G_j = G \setminus \calF_j$. For each pair $(u,v) 
  \in \calK_j$, we check that the minimum $u$-$v$ cut in $G_f$ is at least 1. 
\end{proof}

\begin{remark}
  The above discussion centered around LP relaxations for SNDP variants of problems. 
  These can easily be extended to group connectivity by hitting all sets that 
  separate terminal sets instead of terminal pairs.
\end{remark}

\subsection{Augmentation Framework} 
\mypara{Flex-SNDP:} Suppose 
$G=(V,E), \{s_i, t_i\}_{i \in [h]}$ is an
instance of $(p,q)$-Flex-SNDP.
As above, we assume for simplicity in this discussion that $p_i = p$, $q_i = q$ 
for all $s_i, t_i$ pairs.
We observe that  $(p,0)$-Flex-SNDP instance can be solved via
$2$-approximation to EC-SNDP. Hence, we are interested in $q \ge 1$.
Let $F_1$ be a feasible solution for the $(p,q-1)$-Flex-SNDP instance.
This implies that for any cut $S$ that separates a terminal pair we have
$|\delta_{F_1 \cap \calS}(S)| \ge p$ or $|\delta_{F_1}(S)| \ge p+q-1$.
We would like to augment $F_1$ to obtain a feasible solution to
satisfy the $(p,q)$ requirement. Define a function $f: 2^{|V|} \to \{0, 1\}$ where
$f(S) = 1$ iff (i) $S$ separates a terminal pair and (ii) $|\delta_{F_1 \cap
\calS}(S)| < p$ \emph{and} $|\delta_{F_1}(S)| = p+q-1$.  We call $S$ a 
\emph{violated} cut with respect to $F_1$. Since
$F_1$ satisfies $(p, q-1)$ requirement, if $|\delta_{F_1
\cap \calS}| < p$ it must be the case that $|\delta_{F_1}(S)| \ge p+q-1$. 
The following lemma is simple.

\begin{lemma}
  Suppose $F_2 \subseteq E \setminus F_1$ is a feasible cover for $f$, that is,
  $\delta_{F_2}(S) \ge f(S)$ for all $S$. Then $F_1 \cup F_2$ is a feasible solution 
  to $(p,q)$-Flex-SNDP.
\end{lemma}
The augmentation problem is then to find a min-cost subset of edges to
cover $f$ in $G-F_1$.  The key observation is that the augmentation
problem does not distinguish between safe and unsafe edges and hence
we can rely on traditional connectivity augmentation ideas. Note that if we 
instead tried to augment from $(p-1,q)$ to
$(p,q)$-flex-connectivity, we would still need to distinguish between
safe and unsafe edges. The following lemma shows that the LP relaxation for the 
original instance provides a valid
cut-covering relaxation for the augmentation problem.

\begin{lemma}
\label{lemma:augmentation_lp}
  Let $x \in [0,1]^{|E|}$ be a feasible LP solution for a given
  instance of $(p,q)$-Flex-Steiner. Let $F_1$ be a feasible solution
  that satisfies $(p,q-1)$ requirements for the terminal.  Then, for
  any violated cut $S \subseteq V$ in $(V, F_1)$, we have
  $\sum_{e \in \delta(S) \setminus F_1} x_e \geq 1$.
\end{lemma}
\begin{proof}
  Suppose $S \subseteq V$ is violated in $(V,F_1)$, then
  $|\delta_{F_1}(S)| = p+q-1$, and $|\delta_{F_1 \cap \calS}(S)| <
  p$. Therefore, there are at least $q$ unsafe edges in
  $\delta_{F_1}(S)$. Let $B \subseteq \delta_{F_1}(S) \cap \calU$ such that $|B| =
  q$. Then, $|\delta_{F_1}(S) \setminus B| \leq p-1$, and since
  $x_e \leq 1$ for all $e \in E$,
  $\sum_{e \in \delta_{F_1}(S) - B} x_e \leq p-1$. By the first LP
  constraint, $\sum_{e \in \delta(S) - B} x_e \geq p$. Therefore,
  $\sum_{e \in \delta(S) \setminus F_1} x_e \geq 1$ as desired.
\end{proof}

\mypara{Bulk-SNDP:} Adjiashvili et al. \cite{AdjiashviliSZ15} also define a corresponding augmentation 
problem for the bulk robust network design model as follows: given an instance to 
Bulk-SNDP, let $\Omega_\ell = \bigcup_{j \in [m]} \{(F, \calK_j): |F| 
\leq \ell \text{ and } F \subseteq F_j\}.$ Let $H_{\ell} \subseteq E$ 
be a subset of edges that satisfy the constraints defined by the scenarios in 
$\Omega_\ell$. Then, a solution to the augmentation problem from $\ell-1$ to 
$\ell$ is a set of edges $H'$ such that $H_{\ell-1} \cup H'$ satisfies the 
constraints defined by scenarios in $\Omega_{\ell}$. It is not difficult to 
verify that any solution to the original instance $\Omega$ is also a solution 
to any of the augmentation problems, and any solution satisfying all scenarios 
in $\Omega_{k}$ also satisfies all scenarios in $\Omega$, where $k$ is the width 
of the Bulk-SNDP instance. 

\section{Uncrossability-Based Approximation Algorithms}
\label{sec:fgc}

In this section, we prove Theorem~\ref{thm:introfgc} and 
Theorem~\ref{thm:intro-flex-st}. Recall that we are given a graph $G = (V, E)$
with cost function on the edges $c: E \to \R_{\geq 0}$ and a partition of the 
edge set $E$ into safe edges $\calS$ and unsafe edges $\calU$. In the $(p,q)$-FGC 
problem, our goal is to find the cheapest set of edges such that every cut has 
either $p$ safe edges or $p+q$ total edges. The $(p,q)$-Flex-ST problem is 
similar, except that we are also given $s, t \in V$ as part of the input and 
we focus only on cuts separating $s$ from $t$. 

We start by providing some necessary background on uncrossable/ring families 
and submodularity of the cut function. We then prove a simple $O(q)$-approximation 
for $(2,q)$-FGC by directly applying existing algorithms for covering uncrossable 
functions. Next, we devise a framework for augmentation when the requirement 
function is not uncrossable. Finally, we prove our results for special cases of 
$(p,q)$-FGC and $(p,q)$-Flex-ST using this framework.

\mypara{Uncrossable functions and families:} Uncrossable functions are a general 
class of requirement functions that are an important ingredient in network design 
\cite{WilliamsonGMV95,GoemansW97,GuptaK11,KortsarzN10}.

\begin{definition}
  A function $f: 2^{V} \to \{0, 1\}$ is \emph{uncrossable} if for
  every $A, B \subseteq V$ such that $f(A) = f(B) = 1$, one of the
  following is true: 
  (i) $f(A \cup B) = f(A \cap B) = 1$, 
  (ii)  $f(A - B) = f(B - A) = 1$.
  A family of cuts $\calC \subset 2^V$ is an uncrossable \emph{family} 
  if the indicator function
  $f_{\calC}:2^V \rightarrow \{0,1\}$ with $f(S) = 1$ iff $S \in \calC$, 
  is uncrossable.
\end{definition}

For a graph $G = (V, E)$, a requirement function $f: 2^{V} \to \{0, 1\}$, 
and a subset of edges $A \subseteq E$, we say a set $S \subseteq V$ is violated 
with respect to $A, f$ if $f(S) = 1$ and $\delta_A(S) = \emptyset$. The following 
important result gives a 2-approximation algorithm for the problem of covering an 
uncrossable requirement function.

\begin{theorem}[\cite{WilliamsonGMV95}]
\label{thm:uncrossable}
  Let $G=(V,E)$ be an edge-weighted graph and let $f:2^V \rightarrow \{0,1\}$ 
  be an uncrossable function. Suppose there is an efficient oracle that for any 
  $A \subseteq E$ outputs all the
  minimal violated sets of $f$ with respect to $A$. Then there is an efficient 
  $2$-approximation for the
  problem of finding a minimum cost subset of edges that covers $f$.
\end{theorem}

A special case of uncrossable family of sets is a \emph{ring family}.
We say that an uncrossable family $\calC \subseteq 2^V$ is a
ring family if the following conditions hold: (i) if $A,B \in \calC$
and $A,B$ properly intersect\footnote{$A, B$ properly intersect if 
$A \cap B \neq \emptyset$ and $A-B, B-A \neq \emptyset$.} then 
$A \cap B$ and $A \cup B$ are in
$\calC$ and (ii) there is a unique minimal set in $\calC$.  We observe
that if $\calC$ is an uncrossable family such that there is a
vertex $s$ contained in every $A \in \calC$ then $\calC$ is
automatically a ring family.  Theorem~\ref{thm:uncrossable} can be
strengthened for this case. There is an optimum algorithm to find a min-cost
cover of a ring family --- see \cite{Nutov10,Nutov12,Frank1979}.

In order to use
Theorem~\ref{thm:uncrossable} in the augmentation framework, we need to be able 
efficiently find all the
minimal violated sets of the family. As above, we let $F_1$ denote a feasible 
solution for the $(p,q-1)$-Flex-SNDP instance. For any fixed $p,q$, we can
enumerate all minimal violated sets in $n^{O(p+q)}$ time by trying all
possible subsets of $p+q-1$ edges in $F_1$. In the context of $(p,q)$-FGC, the 
total number of violated cuts in the augmentation problem is bounded by $O(n^4)$. 
See \cite{BoydCHI22} and the proof of Lemma~\ref{lemma:lp_polytime_fgc} for details.

For the following sections on $(p, q)$-FGC, we let $\calC$ denote the family of 
violated cuts. Note that such families are symmetric, since 
$\delta(S) = \delta(V - S)$. For any two sets $A, B \in \calC$, if 
$A \cup B = V$ then by symmetry, $V - A, V - B \in \calC$. In this case,
 $V - A = B - A$ and $V - B = A - B$, so $A$ and $B$ uncross.
Therefore, when proving uncrossability of $A$ and $B$, we assume without loss 
of generality that $(A \cup B) \neq V$.

\mypara{Submodularity and posimodularity of the cut function:} It is
well-known that the cut function of an undirected graph is symmetric
and submodular.  Submodularity implies that for all
$A, B \subseteq V$, 
$|\delta(A)| + |\delta(B)| \ge |\delta(A \cap B)| + |\delta(A \cup B)|$.
Symmetry and submodularity also implies posimodularity: for all $A, B
\subseteq V$, $|\delta(A)| + |\delta(B)| \ge |\delta(A - B)| + |\delta(B - A)|$.

\subsection{An $O(q)$-approximation for $(2, q)$-FGC}
\label{subsec:2q_fgc}
The following lemma shows that the augmentation problem for increasing 
flex-connectivity from $(2,q-1)$ to $(2,q)$ for any $q \ge 1$ corresponds to 
covering an uncrossable function. 

\begin{lemma}
\label{lemma:fgc_2q}
  The set of all violated cuts when augmenting from $(2, q-1)$-FGC to $(2, q)$-FGC is 
  uncrossable.
\end{lemma}
\begin{proof}
  Let $F \subseteq E$ be a feasible solution to $(2, q-1)$-FGC. Suppose $A$, $B$ are 
  two violated cuts that properly intersect. Since $A$ is a violated cut, 
  $\delta_F(A)$ contains at most $1$ safe edge and at most $q+1$ total edges, 
  similarly $\delta_F(B)$. Since $F$ satisfies the requirements of $(2, q-1)$-FGC, 
  $\delta_F(A)$ and $\delta_F(B)$ must each have exactly $q+1$ total edges.

  First, suppose $|\delta_{F \cap \calS}(A \cap B)| < 2$
  and $|\delta_{F \cap \calS}(A \cup B)| < 2$. Since $F$ is feasible for
  $(2,q-1)$-FGC, $|\delta_F(A \cap B)| \geq q+1$, and
  $|\delta_F(A \cup B)| \geq q+1$. By submodularity of the cut function,
  $|\delta_F(A \cap B)| + |\delta_F(A \cup B)| \leq |\delta_F(A)| +
  |\delta_F(B)| = 2(q+1)$. Therefore,
  $|\delta_F(A \cap B)| = |\delta_F(A \cup B)| = q+1$. Since both $\delta_F(A \cup B)$ and
  $\delta_F(A \cap B)$ have exactly $q+1$ edges and at most one safe edge, 
  $A \cup B$ and $A \cup B$ are violated. This implies $A$ and $B$ uncross. 

  By the same reasoning, using posimodularity of the cut function, we
  can argue that at least one of $\delta_F(A - B)$ and $\delta_F(B - A)$
  must have at least $2$ safe edges otherwise both $A-B,B-A$ are
  violated which implies that $A,B$ uncross. Suppose 
  $|\delta_{F \cap \calS}(A - B)| < 2$ and
  $|\delta_{F \cap \calS}(B - A)| < 2$. Since $F$ is feasible for
  $(2,q-1)$-FGC, $|\delta_F(A - B)| \geq q+1$, and
  $|\delta_F(B - A)| \geq q+1$. By posimodularity of the cut function,
  $|\delta_F(A - B)| + |\delta_F(B - A)| \leq |\delta_F(A)| +
  |\delta_F(B)| = 2(q+1)$. Therefore,
  $|\delta_F(A - B)| = |\delta_F(B - A)| = q+1$. Since $A - B$ and $B - A$ 
  are crossed by exactly $q+1$ edges and at most one safe edge, 
  they are both violated cuts.

  Thus, we can assume that one of $\delta_F(A \cap B)$ and
  $\delta_F(A \cup B)$ has $2$ safe edges and one of $\delta_F(A - B)$
  and $\delta_F(B - A)$ has $2$ safe edges.  Assume that 
  $|\delta_{F \cap \calS}(A \cap B)| \geq 2$ and
  $|\delta_{F \cap \calS}(A - B)| \geq 2$; the other cases are
  similar. Recall that $\delta_F(A)$ and $\delta_F(B)$ each have at most
  one safe edge. Therefore,
  $|\delta_{F \cap \calS}(A \cap B) \cup \delta_{F \cap \calS}(A - B)|
  \leq 2$, so there must be two safe edges in
  $\delta_F(A - B) \cap \delta_F(A \cap B)$, i.e. both safe edges must
  be shared between the boundaries of $A - B$ and $A \cap B$. However
  this implies that $\delta_F(B)$ has two safe edges contradicting that
  $B$ was violated.
\end{proof}

The preceding lemma yields a $2(q+1)$-approximation for $(2,q)$-FGC
as follows. We start with a $2$-approximation for $(2,0)$-FGC that can
be obtained by using an algorithm for $2$-ECSS. Then for we augment in
$q$-stages to go from a feasible solution to $(2,0)$-FGC to
$(2,q)$-FGC. The cost of augmentation in each stage is at most $\opt$
where $\opt$ is the cost of an optimum solution to $(2,q)$-FGC. We can
use the known $2$-approximation algorithm in each augmentation stage
since the family is uncrossable. Recall from Section~\ref{sec:prelim}
that the violated cuts can be enumerated in polynomial time, and
hence the primal-dual $2$-approximation for covering an uncrossable
function can be implemented in polynomial-time. This leads to the
claimed approximation and running time.

\subsection{Identifying Uncrossable Subfamilies}
\label{subsec:uncrossable_families_algo}
We have seen that the augmentation problem from $(2,q-1)$-FGC to
$(2,q)$-FGC leads to covering an uncrossable function. Boyd et
al. \cite{BoydCHI22} showed that augmenting from $(p,0)$-FGC to $(p,1)$-FGC
also leads to an uncrossable function for any $p \ge 1$. However this approach 
fails for most cases of augmenting from $(p,q-1)$ to
$(p,q)$. We give an example for $(3,1)$ to $(3,2)$ in
Figure~\ref{fgc_32_counterex}.

\begin{figure}
  \begin{center}
    \includegraphics[width = 0.3\linewidth]{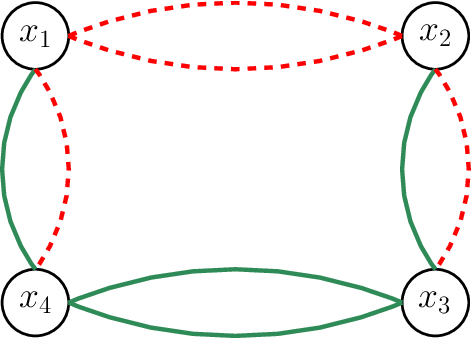}
  \end{center}
  \caption{Violated cuts when augmenting from 
  $(3,1)$-FGC to $(3,2)$-FGC are not uncrossable. Red dashed edges are unsafe, 
  green solid edges are safe. Let $A = \{x_1, x_2\}, B = \{x_2, x_3\}$. 
  $A$ and $B$ are
  violated, since they each are crossed by two safe and two unsafe edges, but
  $A \cup B = \{x_1, x_2, x_3\}$ and $B - A = \{x_3\}$ are not, since
  they each are crossed by three safe edges.}
  \label{fgc_32_counterex}
\end{figure}

However, in certain cases, we can take a more
sophisticated approach where we consider the violated cuts in a small number of stages. 
In each stage, we choose a subfamily of the violated cuts that is
uncrossable.  In such cases, we can obtain a $2\ell$-approximation
for the augmentation problem, where $\ell$ is the upper bound on the
number of stages.

Suppose we want to augment from $(p, q-1)$ to $(p, q)$ (for either FGC or the Flex-ST 
setting). Let $G = (V, E)$ be the original input graph, and let $F$ be our current
partial solution. Recall that a cut $\emptyset \neq A \subsetneq V$ is violated iff 
$|\delta_F(A)| = p+q-1$, $|\delta_{F\cap \calS}(A)| < p$, and in the Flex-ST case, 
$A$ separates $s$ from $t$. We cover violated cuts in stages. In each stage we 
consider the violated cuts based on the number of safe edges. We begin by covering 
all violated sets with no safe edges, then with one safe edge, and iterate until 
all violated sets are covered. 
This is explained in Algorithm~\ref{algo:augmentation_in_stages} below.

\begin{algorithm}[H]
\caption{Augmenting from $(p,q-1)$ to $(p,q)$ in stages}
\label{algo:augmentation_in_stages}
\begin{algorithmic}[1]
    \State $F' \gets F$
    \For{$i = 0, \dots, p-1$}
        \State $\calC_i \gets \{S : S \text{ is violated and } |\delta_{F' \cap \calS}(S)| = i\}$
        \State $F_i' \gets$ approximation algorithm to cover cuts in $\calC_i$
        \State $F' \gets F' \cup F_i'$
    \EndFor\\
    \Return $F'$
\end{algorithmic}
\end{algorithm}

\subsection{Approximating $(p, q)$-FGC for $q \leq 4$}
\label{subsec:fgc_p4}
In this section, we show that the above approach works to augment
from $(p,q-1)$-FGC to $(p,q)$-FGC whenever $q \le 3$ and also for $q=4$ when
$p$ is even.
The only unspecified part Algorithm~\ref{algo:augmentation_in_stages} is to cover 
cuts in $\calC_i$ in the $i$'th stage.  If we can prove that $\calC_i$ forms
an uncrossable family then we can obtain a $2$-approximation in each
stage. First, we prove a generic and useful lemma regarding cuts
in $\calC_i$. We let $F_i \subseteq E$ denote the set of edges $F'$ at the start of 
iteration $i$. In other words, $F_i$ is a set of edges such that for all 
$\emptyset \neq A \subsetneq V$, if $|\delta_{F_i}(A)| = p+q-1$, then
$|\delta_{F_i \cap \calS}(A)| \geq i$.

\begin{lemma}
\label{lemma:fgc_structure}
  Fix an iteration $i \in \{0, \dots, p-1\}$. Let $\calC_i$ be as defined in 
  Algorithm~\ref{algo:augmentation_in_stages}. Then, if $A, B \in \calC_i$ and
  \begin{enumerate}
      \item $|\delta_{F_i}(A \cap B)| = |\delta_{F_i}(A \cup B)| = p+q-1$, or
      \item $|\delta_{F_i}(A - B)| = |\delta_{F_i}(B - A)| = p+q-1$
  \end{enumerate}
  then $A$ and $B$ uncross, i.e. $A \cap B, A \cup B \in \calC_i$ or 
  $A - B, B - A \in \calC_i$.
\end{lemma}
\begin{proof}
  By submodularity and posimodularity of the cut function, 
  $|\delta_{F_i}(A \cap B)| + |\delta_{F_i}(A \cup B)|$ and 
  $|\delta_{F_i}(A - B)| + |\delta_{F_i}(B - A)|$ are both at most $2(p+q-1)$.
  First, suppose $|\delta_{F_i}(A \cap B)| = |\delta_{F_i}(A \cup B)| = p+q-1$. 
  By submodularity of the cut function applied to the safe edges $F_i \cap \cal S$,
  \[2i = |\delta_{F_i \cap \calS}(A)| + |\delta_{F_i \cap \calS}(B)| 
  \geq |\delta_{F_i \cap \calS}(A \cup B)| + |\delta_{F_i \cap \calS}(A \cap B)|.\]
  Since we are in iteration $i$, neither $\delta_{F_i}(A \cup B)$ nor 
  $\delta_{F_i}(A \cap B)$ can have less than $i$ safe edges. 
  Therefore, $|\delta_{F_i \cap \calS}(A \cup B)| = 
  |\delta_{F_i \cap \calS}(A \cap B)| = i$. Then, $A \cap B, A \cup B \in \calC_i$, so 
  $A$ and $B$ uncross.
  In the other case, where $|\delta_{F_i}(A - B)| = |\delta_{F_i}(B - A)| = p+q-1$, 
  we can use posimodularity of the cut function to show that 
  $|\delta_{F_i \cap \calS}(A - B)| = |\delta_{F_i \cap \calS}(B - A)| = i$, 
  so $A - B, B - A \in \calC_i$.
\end{proof}

Note that the preceding lemma holds for the high-level approach. Now we focus on
cases where we can prove that $\calC_i$ is uncrossable.

\begin{lemma}
\label{lemma:fgc_p3_uncrossable}
  Fix an iteration $i \in \{0, \dots, p-1\}$. Let $\calC_i$ be as defined in 
  Algorithm~\ref{algo:augmentation_in_stages}. Then, for $q \leq 3$, $\calC_i$ is 
  uncrossable.
\end{lemma}
\begin{proof}
  Let $A, B \in \calC_i$, and suppose for the sake of contradiction that they do not
  uncross. By submodularity and posimodularity of the cut function, 
  $|\delta_{F_i}(A \cap B)| + |\delta_{F_i}(A \cup B)|$ and 
  $|\delta_{F_i}(A - B)| + |\delta_{F_i}(B - A)|$ are both at most $2(p+q-1)$.
  Thus by Lemma~\ref{lemma:fgc_structure}, at least one of
  $\delta_{F_i}(A \cap B)$ and $\delta_{F_i}(A \cup B)$ must have at
  most $p+q-2$ edges, and the same holds for $\delta_{F_i}(A - B)$ and
  $\delta_{F_i}(B - A)$. Without loss of generality, suppose
  $\delta_{F_i}(A \cap B)$ and $\delta_{F_i}(A - B)$ each have at
  most $p+q-2$ edges. By the assumptions on $F_i$, they must both have
  at least $p$ safe edges, hence they each have at most $q-2$ unsafe
  edges. Note that
  $\delta_{F_i}(A) \subseteq \delta_{F_i}(A - B) \cup \delta_{F_i}(A
  \cap B)$, hence $\delta_{F_i}(A)$ can have at most $2(q-2)$ unsafe
  edges. When $q \leq 3$, $2(q-2) < q$, which implies that
  $\delta_{F_i}(A)$ has strictly more than $p-1$ safe edges, a
  contradiction. Notice that
  $\delta_{F_i}(A) \subseteq \delta_{F_i}(B - A) \cup \delta_{F_i}(A
  \cup B)$,
  $\delta_{F_i}(B) \subseteq \delta_{F_i}(A - B) \cup \delta_{F_i}(A
  \cup B)$, and
  $\delta_{F_i}(B) \subseteq \delta_{F_i}(B - A) \cup \delta_{F_i}(A
  \cap B)$; therefore the same argument follows regardless of which
  pair of sets each have strictly less than $p+q-2$ edges.
\end{proof}

\begin{corollary}
\label{cor:fgc_p3}
  For any $p \ge 2$ there is a $(2p+4)$-approximation for $(p,2)$-FGC
  and a $(4p+4)$-approximation for $(p,3)$-FGC.
\end{corollary}
\begin{proof}
  Via \cite{BoydCHI22} we have a $4$-approximation for $(p,1)$-FGC. We
  can start with a feasible solution $F$ for $(p,1)$-FGC and use
  Algorithm~\ref{algo:augmentation_in_stages} to augment from $(p,1)$ to
  $(p,2)$; each stage can be approximated to within a factor of $2$
  since $\calC_i$ is an uncrossable family. Since there are $p$ stages
  and the cost of each stage can be upper bounded by the cost of
  $F^*$, an optimum solution to the $(p,2)$ problem, the total cost of
  augmentation is $2p \cdot \opt$. This leads to the desired
  $(2p+4)$-approximation for $(p,2)$-FGC. For $(p,3)$-FGC we can augment
  from $(p,2)$ to $(p,3)$ paying an additional cost of $2p \cdot \opt$. This
  leads to the claimed $(4p+4)$-approximation. Following the discussion
  in Section~\ref{sec:prelim}, covering the uncrossable families that
  arise in Algorithm~\ref{algo:augmentation_in_stages} can be done in
  polynomial time.
\end{proof}

Can we extend the preceding lemma for $q = 4$? It turns out that it
does work when $p$ is even but fails for odd $p \ge 3$.

\begin{lemma}
\label{lemma:fgc_p4_uncrossable}
  Fix an iteration $i \leq p-2$. Let $\calC_i$ be as
  defined in Algorithm~\ref{algo:augmentation_in_stages}. 
  Then, for $q = 4$, $\calC_i$ is uncrossable. Furthermore, if $p$ is an 
  \emph{even} integer, $\calC_{p-1}$ is uncrossable.
\end{lemma}
\begin{proof}
  Let $A, B \in \calC_i$, so $\delta_{F_i}(A)$ and $\delta_{F_i}(B)$
  both have exactly $i$ safe and $p+3-i$ unsafe edges. Suppose for the
  sake of contradiction that they do not uncross. 
  By submodularity and posimodularity of the cut function, 
  $|\delta_{F_i}(A \cap B)| + |\delta_{F_i}(A \cup B)|$ and 
  $|\delta_{F_i}(A - B)| + |\delta_{F_i}(B - A)|$ are both at most $2(p+q-1)$.
  Thus by Lemma~\ref{lemma:fgc_structure}, at least one of $\delta_{F_i}(A \cap B)$
  and $\delta_{F_i}(A \cup B)$ must have at most $p+2$ edges, and the
  same holds for $\delta_{F_i}(A - B)$ and $\delta_{F_i}(B -A)$. 
  Without loss of generality, suppose $\delta_{F_i}(A \cap B)$
  and $\delta_{F_i}(A - B)$ each have at most $p+2$ edges. By the
  assumptions on $F_i$, they must both have at least $p$ safe
  edges. 

  Notice that each safe edge on $\delta_{F_i}(A \cap B)$ or
  $\delta_{F_i}(A - B)$ must be in either $\delta_{F_i}(A)$ or
  $\delta_{F_i}(A \cap B) \cap \delta_{F_i}(A - B)$. Let
  $\ell = |\delta_{F_i \cap \calS}(A \cap B) \cap \delta_{F_i \cap
    \calS}(A - B)|$, i.e. the number of safe edges crossing both
  $A \cap B$ and $A - B$. Since $|\delta_{F_i \cap \calS}(A)| = i$,
  \[i = |\delta_{F_i \cap \calS}(A)| 
  \geq |\delta_{F_i \cap \calS}(A \cap B)| + |\delta_{F_i \cap \calS}(A - B)| - 
  2\ell \geq 2p - 2\ell,\]
  implying that $\ell \geq \frac{2p-i}2$. Note that 
  $\delta_{F_i}(A) \subseteq \delta_{F_i}(A \cap B) \cup \delta_{F_i}(A - B)$. 
  Furthermore, if an edge crosses both $A \cap B$ and $A - B$, then it must 
  have one endpoint in each (since they are disjoint sets), and therefore does not 
  cross $A$. Therefore,
  \begin{align*}
      |\delta_{F_i}(A)| \leq |\delta_{F_i}(A \cap B)| + |\delta_{F_i}(A - B)| - 
      2\ell \leq 2(p+2) - 2 \cdot \left \lceil \frac{2p-i}{2} \right \rceil 
      \leq 4+i
  \end{align*}
  When $i < p-1$, this is at most $p+2$. When $i = p-1$ and $p$ is
  even, $2 \cdot \left \lceil \frac{2p-i}{2} \right \rceil = 2 \cdot
  \left \lceil \frac{p+1}{2} \right \rceil = p+2$, and $2(p+2) - p+2 =
  p+2$. In either case, we get $|\delta_{F_i}(A)| \leq p+2$, a
  contradiction to the assumption on $A$.
\end{proof}

The preceding lemma leads to a $(6p+4)$-approximation for $(p, 4)$-FGC
when $p$ is even by augmenting from a feasible solution to
$(p,3)$, since we pay an additional cost of $2p \cdot \opt$. 
The preceding lemma also shows that the bottleneck for odd $p$ is in covering
$\calC_{p-1}$.  We show via an example that the family is indeed not
uncrossable (see Figure \ref{fgc_k4odd_counterex}).
It may be possible to show that $\calC_{p-1}$ separates into a constant
number of uncrossable families leading to an $O(p)$-approximation for
$(p,4)$-FGC for all $p$. The first non-trivial case is when $p=3$. 
Note that in Figure \ref{fgc_k4odd_counterex}, none of the 
cuts $A \cup B, A \cap B, A - B, B - A$ are violated, thus a new 
approach to partitioning the set of violated cuts is necessary.

\begin{figure}
  \begin{center}
  \includegraphics[width = 0.4\linewidth]{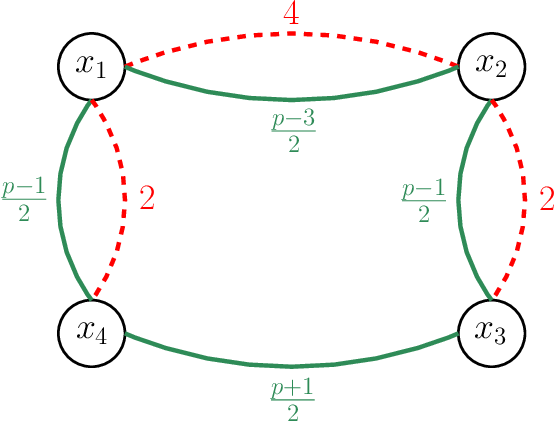}
  \end{center}
  \caption{Violated cuts in augmentation problem from 
  $(p, 3)$-FGC to $(p, 4)$-FGC are not uncrossable when $p$ is odd. 
  The numbers on edges denote the number of parallel edges, red dashed edges
   are unsafe, green solid edges are safe. Let
   $A = \{x_1, x_2\}$, $B = \{x_2, x_3\}$; both are crossed by
  exactly $p-1$ safe edges and $4$ unsafe edges. However,
   $A \cup B$ and $B - A$ are each crossed by $p$ safe edges, and $A \cap B$ and
   $A - B$ are each crossed by $p+4$ total edges.}
  \label{fgc_k4odd_counterex}
\end{figure}

We also demonstrate that preceeding lemmas do not extend for $q=5$, even for
even $p$, by giving an example where two violated cuts are not uncrossable
when augmenting from $(4, 4)$-FGC to $(4, 5)$-FGC in
Figure \ref{fgc_k5_counterex}. Note that $A$ and $B$ each only have 3 safe edges
crossing them, so the argument that $\calC_i$ is uncrossable for
$i < p-1$ fails. 

\begin{figure}
  \begin{center}
    \includegraphics[width = 0.4\linewidth]{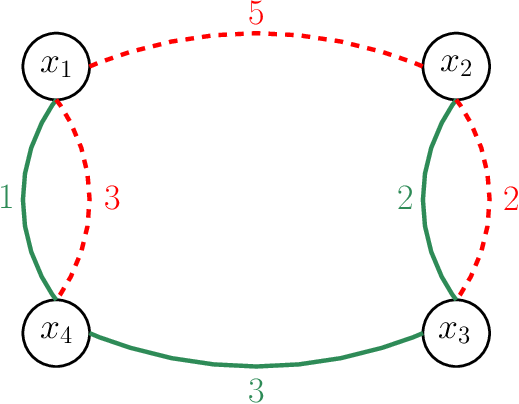}
  \end{center}
  \caption{Violated cuts in augmentation problem from 
  $(4,4)$-FGC to $(4,5)$-FGC are not uncrossable. The numbers on edges 
  denote the number of parallel edges. Red dashed edges are unsafe, 
  green solid edges are safe. Let $A = \{x_1, x_2\}$, let
  $B = \{x_2, x_3\}$; both are in $\calC_3$. $A \cup B$ and $B - A$ each have at least $4$ safe
  edges, and $A \cap B$ and $A - B$ each have $9$ total edges, so none
  are violated.}
  \label{fgc_k5_counterex}
\end{figure}

\subsection{An $O(1)$-Approximation for Flex-ST}
\label{subsec:st}
In this section, we provide a constant factor approximation for 
$(p, q)$-Flex-ST for all fixed $p,q$ that satisfy $p+q > pq/2$. We follow 
the general approach outlined in \ref{subsec:uncrossable_families_algo} with 
some modifications. In particular, we start with a stronger set of edges $F$ 
than in the FGC case above. Let $E' \subseteq E$ denote a feasible solution 
to $(p, q-1)$-Flex-ST.

Recall from Section~\ref{sec:intro} the capacitated network design problem, 
in which each edge has an integer capacity $u_e \geq 1$. Consider an instance 
of the $(p(p+q))$-Cap-ST problem on $G$ where every safe edge is given a capacity 
of $p+q$ and every unsafe edge is given a capacity of $p$. Our goal is to find 
the cheapest set of edges that support a flow of $(p(p+q))$ from $s$ to $t$. 
It is easy to see that any solution to $(p,q)$-Flex-ST is also a feasible 
solution for this capacitated problem: every $s$-$t$ cut either has at least 
$p$ safe edges or at least $p+q$ total edges, and either case gives a capacity 
of at least $p(p+q)$. As mentioned in Section~\ref{sec:intro}, there exists 
a $\max_{e}(u_e) = (p+q)$-approximation for this problem. 
Let $E'' \subseteq E$ be such a solution, and note that 
$\cost(E'') \leq (p+q) \cdot \opt$, where $\opt$ denotes the cost of 
an optimal solution to $(p,q)$-Flex-ST.

\begin{figure}
  \begin{center}
    \includegraphics[width = 0.4\linewidth]{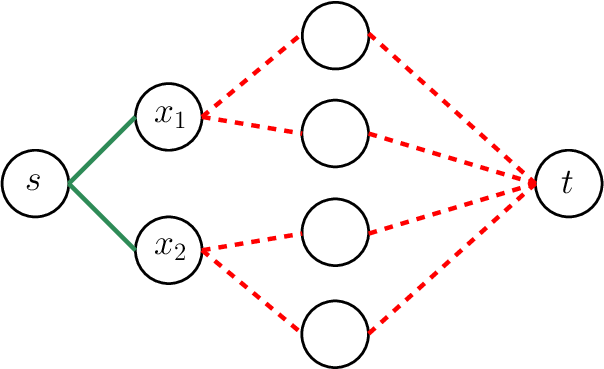}
  \end{center}
  \caption{Example where $\calC_1$ is not uncrossable when augmenting from 
  $(2,1)$ to $(2,2)$-Flex-ST. Red dashed edges are unsafe, while green solid 
  edges are safe. Let $A = \{s, x_1\}, B = \{s, x_2\}$. $A$ and $B$ are both 
  crossed by exactly one safe and two unsafe edges, but 
  $A \cup B$ is crossed by two safe edges and $A \cap B$ is crossed by 
  four unsafe edges. $A - B$ and $B - A$ are not $s$-$t$ cuts.}
  \label{figure:st_22_counterex}
\end{figure}

Let $F = E' \cup E''$. We redefine $\calC$ to limit ourselves to the set of 
violated cuts containing $s$, i.e. $\calC = \{A \subset V: s \in A, t \notin A, 
|\delta_{F \cap \calS}(A)| < p, |\delta_{F \cap \calU}(A)| = p+q-1\}$. 
By symmetry, it suffices to only consider cuts containing $s$, since covering a set
also covers its complement. Following the discussion in 
Section~\ref{subsec:uncrossable_families_algo}, we use 
Algorithm~\ref{algo:augmentation_in_stages} to cover violated cuts in stages 
based on the number of safe edges. However, unlike the spanning case, the 
sets $\calC_i$ are not uncrossable in the single pair setting, even for 
$(2,2)$-Flex-ST (see Figure \ref{figure:st_22_counterex}). In this case, we aim 
to further partition $\calC_i$ into subfamilies that we can cover efficiently. 

For the rest of the proof, we let $F_i \subseteq E$ denote the set of edges $F'$ 
at the start of iteration $i$. In other words, $F_i$ is a set of edges such that 
for all cuts $A$ separating $s$ from $t$, if $|\delta_{F_i}(A)| = p+q-1$, then 
$|\delta_{F_i \cap \calS}(A)| \geq i$. We begin with a structural observation. 

\begin{lemma}
\label{lemma:st_structure}
  Fix an iteration $i \in \{0, \dots, p-1\}$. Let $\calC_i$ be as defined in 
  Algorithm~\ref{algo:augmentation_in_stages}. Let $A, B \in \calC_i$. Then, either
  \begin{enumerate}
    \item $A \cup B, A \cap B \in \calC_i$, or
    \item $\max(\delta_{F_i \cap \calS}(A \cap B), 
\delta_{F_i \cap \calS}(A \cup B)) \geq p$.
  \end{enumerate}
\end{lemma}
\begin{proof}
  First, note that if $A$ and $B$ separate $s$ from $t$, then $A \cup B$ and 
  $A \cap B$ both do as well. First, suppose $|\delta_{F_i}(A \cap B)| < p+q-1$. 
  Since $F \subseteq F_i$, any $s$-$t$ separating cut with less than $p+q-1$ edges 
  must have at least $p$ safe edges. Thus $|\delta_{F \cap S}(A \cap B)| \geq p$. 
  The argument is similar for the case where $|\delta_{F_i}(A \cup B)| < p + q -1$.

  By submodularity of the cut function, 
  $|\delta_{F_i}(A \cap B)| + |\delta_{F_i}(A \cup B)| \leq |\delta_{F_i}(A)| + 
  |\delta_{F_i}(B)| = 2(p+q-1).$ Therefore, if neither $\delta_{F_i}(A \cap B)$ nor 
  $\delta_{F_i}(A \cup B)$ have  strictly less than $p+q-1$ edges, they must both 
  have exactly $p+q-1$ edges. By submodularity of the cut function applied to 
  $F_i \cap \calS$, $|\delta_{F_i \cap \calS}(A \cap B)| + 
  |\delta_{F_i \cap \calS}(A \cup B)| \leq |\delta_{F_i \cap \calS}(A)| + 
  |\delta_{F_i \cap \calS}(B)| = 2i.$
  By definition of $F_i$, neither $\delta_{F_i}(A \cap B)$ nor $\delta_{F_i}(A \cup B)$ 
  can have less than $i$ safe edges. Therefore, 
  $|\delta_{F_i \cap \calS}(A \cap B)| = |\delta_{F_i \cap \calS}(A \cup B)| = i$, 
  so $A \cap B, A \cup B \in \calC_i$.
\end{proof}

Consider a flow network on the graph $(V, F_i)$ with safe edges given a capacity 
of $(p+q)$ and unsafe edges given a capacity of $p$. Since $E'' \subseteq F_i$, 
$F_i$ satisfies the $p(p+q)$-Cap-ST requirement. Therefore, the minimum 
capacity $s$-$t$ cut and thus the maximum $s$-$t$ flow value is at least 
$p(p+q)$. Since capacities are integral, there is some integral max flow $f$. 
By flow decomposition, we can decompose $f$ into a set $\calP$ of $|f|$ paths, 
each carrying a flow of 1, and we can find $\calP$ in polynomial time.

For each $\calQ \subseteq \calP$ where $|\calQ| = i$, we define a subfamily of 
violated cuts $\calC_i^{\calQ}$ as follows. Let $\calQ = \{P_1, \dots, P_i\}$. 
Then, $A \in \calC_i$ is in $\calC_i^{\calQ}$ iff there exist distinct 
edges $e_1, \dots, e_i \in \calS$ satisfying:
\begin{enumerate}
  \item $\forall j \in [i]$, $\delta_{F_i}(A) \cap P_j = \{e_j\}$,
  \item $\delta_{F_i \cap \calS}(A) = \{e_1, \dots, e_i\}$.
\end{enumerate}
Informally, $\calC_i^{\calQ}$ is the set of all violated cuts that intersect 
the paths of $\calQ$ exactly once and on a distinct safe edge each. 

\begin{lemma}
\label{lemma:st_allcutscovered}
  Suppose $p+q > pq/2$. If $A \in \calC_i$, then there exists some 
  $\calQ \subseteq \calP$, $|\calQ| = i$, such that $A \in \calC_i^{\calQ}$.
\end{lemma}
\begin{proof}
  Let $A \in \calC_i$. Suppose for every $e \in \delta_{F_i \cap \calS}(A)$, 
  there is a path $P_e \in \calP$ such that $e \in P_e$ and for any 
  $e' \in \delta_{F_i}(A)$, $e' \neq e$ implies that $e' \notin P_e$. 
  Then, by the definition above, $A \in \calC_i^{\calQ}$, where 
  $\calQ = \{P_e: e \in \delta_{F_i \cap \calS}(A)\}$. Therefore, it suffices 
  to show that these paths $P_e$ exist. Suppose for the sake of contradiction 
  that there is some safe edge $e^*$ crossing $A$ without a unique flow path; 
  that is, either there is no path in $\calP$ containing $e^*$, or 
  every path in $\calP$ containing 
  $e^*$ also contains some other $e' \in \delta_{F_i}(A)$.

  Since $A \in \calC_i$, $A$ is crossed by exactly $i$ safe edges 
  (each with capacity $p+q$) and $p+q-1-i$ unsafe edges (each with capacity $p$). 
  Since $e^*$ does not contain any unique flow paths, the total number of paths 
  $P \in \calP$ crossing $\delta_{F_i}(A)$ is at most $(i-1)(p+q) + (p+q-1-i)p = 
  p(p+q) + (iq - 2p - q).$

  Recall that $|\calP| \geq p(p+q)$. Since $A$ separates $s$ from $t$, 
  every path in $\calP$ must intersect $\delta_{F_i}(A)$ at least once. 
  Therefore, $iq - 2p - q$ must be non-negative. However, if $2p + 2q > pq$, 
  then $iq - 2p - q = (i+1)q - (2p + 2q) < (i+1)q - pq \leq 0.$ The last 
  inequality follows from the fact that $i \leq p-1$. This gives us our 
  desired contradiction.
\end{proof}

The above lemma shows that 
$\cup_{\calQ \subseteq \calP, |\calQ| = i} \calC_i^{\calQ} = \calC_i$. 
Therefore, it suffices to cover each $\calC_i^{\calQ}$.

\begin{lemma}
\label{lemma:st_ringfamily}
  For any $\calQ \subseteq \calP$, $|\calQ| = i$, $\calC_i^{\calQ}$ is a ring family.
\end{lemma}
\begin{proof}
  Fix $\calQ \subseteq \calP$ such that $|\calQ| = i$, and let 
  $A, B \in \calC_i^{\calQ}$. Let $P \in \calQ$ be arbitrary. Since 
  $A, B \in \calC_i^{\calQ}$, $\exists e_1, e_2$ (not necessarily distinct) 
  such that $\{e_1\} = \delta_{F_i \cap \calS}(A) \cap P$ and 
  $\{e_2\} = \delta_{F_i \cap \calS}(B) \cap P$. 

  Note that $A \cup B$ and $A \cap B$ are both $s$-$t$ cuts, since $A$ and $B$ 
  both separate $s$ from $t$. Since $P$ is an $s$-$t$ path, it must cross both 
  $\delta_{F_i}(A \cup B)$ and $\delta_{F_i}(A \cap B)$ an odd number of times. 
  However, any edge that crosses $A \cup B$ or $A \cap B$ must also cross either 
  $A$ or $B$. Therefore, $P$ can only cross $A \cup B$ and $A \cap B$ through 
  $e_1$ or $e_2$. This implies that $|P \cap \delta_{F_i}(A \cap B)| = 
  |P \cap \delta_{F_i}(A \cup B)| = 1$. Since $P \in \calQ$ was arbitrary, 
  this holds for all $P \in \calQ$, so $A \cup B$ and $A \cap B$ satisfy 
  condition (1) for $\calC_i^{\calQ}$.

  Next, note that neither $e_1$ nor $e_2$ can be in any other $P' \neq P \in \calP$, 
  since that would violate the condition that both $\delta_F(A)$ and 
  $\delta_F(B)$ intersect each path in $\calP$ at a distinct safe edge. 
  Therefore, $A \cup B$ and $A \cap B$ satisfy condition (2) for $\calC_i^{\calQ}$.

  It remains to show that $A \cup B$ and $A \cap B$ are in $\calC_i$. 
  Since $A \cup B$ and $A \cap B$ are each crossed by a distinct safe edge 
  in $\calP$, they must each have at least $i$ safe edges. By submodularity of 
  the cut function, $|\delta_{F \cap S}(A \cup B)| + |\delta_{F \cap S}(A \cap B)| 
  \leq |\delta_{F \cap S}(A)| + |\delta_{F \cap S}(B)| = 2i$, so they must each 
  have exactly $i$ safe edges. By Lemma~\ref{lemma:st_structure}, since $i < p$, 
  $A \cup B$ and $A \cap B$ must both be in $\calC_i$, as desired.
\end{proof}

We combine the above lemmas to obtain a constant factor approximation for 
covering $\calC_i$.
\begin{lemma}
\label{lemma:st_ci_approx}
  Suppose $p+q > pq/2$ and $p,q$ are fixed. Then, there exists an algorithm that 
  runs in $n^{O(p+q)}$ time to cover all cuts in $\calC_i$ with cost at most 
  $\binom{p^2+pq}{i} \cdot \opt$.
\end{lemma}
\begin{proof}
  We start by finding $\calP$ as described above. For each $\calQ \subseteq \calP$ 
  with $|\calQ| = i$, we can solve the augmentation problem of finding the 
  minimum cost subset $F_\calQ \subseteq E \setminus F_i$ such that 
  $\delta_{F_\calQ}(A) \geq 1$ for all $A \in \calC_i^{\calQ}$. 
  Since $\cup_{\calQ \subseteq \calP, |\calQ| = i} \calC_i^{\calQ} = \calC_i$, 
  $F_i = \cup_{\calQ \subseteq \calP, |\calQ| = i} F_\calQ$ covers all cuts 
  in $\calC_i$. There are $\binom{p(p+q)}i$ such subsets $\calQ$. For each, 
  we can find the minimum violated cuts in $n^{O(p+q)}$ time, since each 
  violated cut in $\calC_i$ has exactly $p+q-1$ edges on its boundary. 

  From the discussion in Section~\ref{sec:prelim}, since each $\calC_i^{\calQ}$ 
  is a ring family, the corresponding augmentation problems can be solved exactly, 
  so $\cost(F_{\calQ}) \leq \opt$. Thus, $\cost(F_i) \leq \binom {p(p+q)}{i} 
  \cdot \opt$.
\end{proof}

The above lemma gives us Theorem~\ref{thm:intro-flex-st} as a corollary. 
At the beginning of each augmentation step, before running 
Algorithm~\ref{algo:augmentation_in_stages}, we compute a solution to the 
$(p(p+q))$-Cap-ST problem, which we can do with cost at most $(p+q)\cdot\opt$. 
Summing over $q$ augmentation iterations gives us the desired $(p+q)^{O(p)}$ 
approximation ratio.

\begin{remark} 
  The approximation factor in Theorem~\ref{thm:intro-flex-st} 
  can be optimized slightly. For example, the algorithm we describe gives a 
  $5$-approximation for $(2,2)$-Flex-ST (see Appendix \ref{sec:st22}).
  We omit the details of this optimization for other values of $(p,q)$ in this 
  paper and instead focus on showing constant factor for fixed $p,q$.
\end{remark}

\section{Approximating Non-Uniform SNDP Problems}
\label{sec:sndp_cutcover}

In this section we prove Theorem \ref{thm:bulk-sndp}, obtaining a polylogarithmic 
approximation algorithm for Bulk-SNDP. We use the augmentation framework 
described in Section \ref{sec:prelim}. Recall that given an instance 
$G, \Omega$ to Bulk-SNDP, we let 
$\Omega_{\ell} = \bigcup_{j \in [m]} \{(F, \calK_j): |F| \leq \ell 
\text{ and } F \subseteq F_j\}.$ The augmentation problem from $\ell-1$ 
to $\ell$ takes a partial solution $H$ satisfying all scenarios in 
$\Omega_{\ell-1}$ and attempts to augment $H$ with an additional set 
$H' \subseteq E \setminus H$ 
such that $H \cup H'$ satisfies all scenarios in $\Omega_{\ell}$.

In the general case of Bulk-SNDP, unlike the special cases discussed
in Section \ref{sec:fgc}, the family of violated cuts do not seem to
have any good structure that we can exploit. Furthermore, there could
be exponentially many violated cuts, even for small $k$. Thus we use a
new cut covering tool given by \cite{gupta2009online}.  At a high
level, the idea is to cover violating \emph{edge} sets instead of
\emph{vertex} sets.  Let $H$ be an existing partial solution
satisfying $\Omega_{\ell-1}$.  We call $F \subseteq E$,
$|F| \leq \ell$ a \emph{violating edge set} if there exists $j$ such
that $(F, \calK_j) \in \Omega_{\ell}$ and $H \setminus F$ disconnects
some pair $u, v \in \calK_j$.  Since $|E| \leq n^2$, there are at most
$O(n^{2\ell})$ violating edge sets. We aim to \emph{cover} $F$ with respect 
to $(u,v)$ by buying a path
in $E \setminus F$ from the component of $H \setminus F$ containing
$u$ to the component of $H \setminus F$ containing $v$.
The following is a simple observation.

\begin{claim}
	$H' \subseteq E \setminus H$ is a feasible solution to the augmentation problem 
	iff for each violating edge set $F$ and every terminal pair $(u,v)$ disconnected 
	in $H \setminus F$, $H'$ contains a path covering $F$ with respect to $(u,v)$.
\end{claim}

Unfortunately, covering edge sets is more complicated than covering vertex 
sets, since $H \setminus F$ may
have many components and there could be exponentially many paths between components of 
$H \setminus F$. The key 
idea of \cite{gupta2009online} is a tree-based cost-sharing lemma that
allows us to consider only polynomially many paths with a polylogarithmic loss
in the approximation ratio. We begin by providing background on this machinery 
before giving the algorithm and analysis for Bulk-SNDP.

\subsection{Cut-Cover Lemma and Tree Embeddings}

In \cite{gupta2009online}, Gupta, Krishnaswamy, and Ravi consider
the online version of the SNDP problem, where they are initially given a graph 
$G = (V, E)$,
and terminal pairs to be connected arrive in an online fashion. They obtain 
an $O(k \log^3 n)$ approximation via the augmentation framework. For any 
tree $T$, we let $P_T(u,v)$ denote the unique tree path between $u$ and $v$. 
If $T$ has the same vertex set as $G$, we use $P_T(e)$ interchangeably with 
$P_T(u,v)$ for an edge $e = uv \in E(G)$.
The key idea in \cite{gupta2009online} is to show that if we have a 
spanning tree $T$ and a partial solution $H$ that contains relevant tree paths, 
then we can cover all violated edge sets using only the 
fundamental cycles of the tree, i.e. $\{e\} \cup P_T(e)$ for $e \in E$. 
It is not difficult to see that this idea can be extended 
beyond SNDP to non-uniform settings as well. We summarize a slightly generalized
version of their results in the following lemma, referred to as the 
\emph{cut-cover lemma}.

\begin{lemma}[\cite{gupta2009online}]
\label{lem:cutcover}
	Consider a graph $G = (V, E)$, let $s, t \in V$ and let $T$ be any spanning 
	tree of $G$. Let $H \subseteq G$ be a subgraph of 
	$G$ that contains $P_T(s, t)$ and let $F \subseteq E(H)$ be any minimal set of 
	edges of $H$ such that $s$ and $t$ are disconnected in $H \setminus F$.
	We denote by $L$ and $R$ the connected components of $H \setminus F$ containing 
	$s$ and $t$ respectively. Suppose for some $E^* \subseteq E$ that 
	$s$ and $t$ are connected in $E^* \setminus F$. Then, 
	$\exists e \in E^*$ such that the fundamental cycle 
	$\{e\} \cup P_T(e) \setminus F$ connects some vertex in $L$ to some 
	vertex in $R$.
\end{lemma}

In order to use the cut-cover lemma to obtain an approximation algorithm, one
must obtain a good spanning tree. In particular, we will need a tree in which 
the cost $P_T(e)$ is not too much more than $c(e)$. We use the following result on 
probabilistic distance-preserving spanning tree embeddings.

\begin{lemma}[\cite{abraham2008nearly,elkin2005lower}]
\label{lem:spanningtree_embeddings}
	Given a graph $G = (V, E)$ with non-negative edge costs $c(e)$, there exists 
	a distribution $\calT$ over spanning trees $T = (V, E_T)$, $E_T \subseteq E$
	such that for all $u, v \in V$, $\E[d_T(u,v)] \leq 
	O(\log n \log \log n) d_G(u,v)$. 
	The metrics $d_T$ and $d_G$ are defined by the shortest paths in $T$ and $G$ 
	respectively according to $c(e)$.
\end{lemma}

\subsection{Algorithm for Augmentation Problem}

We show how to apply the cut cover lemma of \cite{gupta2009online} to solve the 
Bulk-SNDP augmentation problem from $\ell-1$ to $\ell$, where $H_{\ell-1}$ is the partial 
solution for $\ell-1$.
Let $T$ be a tree sampled from the distribution 
$\calT$ given by Lemma \ref{lem:spanningtree_embeddings}.
Let $H_P = \cup_{j \in [m]} \cup_{(u, v) \in \calK_j} P_T(u,v)$. Let 
$H$ denote $H_{\ell-1} \cup H_P$.
We show that we can solve the 
augmentation problem via a reduction to Hitting Set. 

In the Hitting Set problem, we are given a universe $\calU$ of elements along
with a collection of sets $\calS = \{S_1, \dots, S_m\}$, where each 
$S_i \subseteq \calU$. Each element $e$ in $\calU$ has a cost $c(e)$. The goal is to 
find the minimum cost subset of elements $X \subseteq \calU$ such that 
$X \cap S_i \neq \emptyset$ for all $i \in [m]$. We say an element $x_j \in \calU$
\emph{hits} a set $S_i$ if $x_j \in S_i$. It is well known that 
Hitting Set admits an $O(\log m)$ approximation in time 
$\poly(n,m)$. We define the following Hitting Set instance:
\begin{itemize}
	\item \emph{Sets:} Each pair $(F, uv)$ corresponds to a set,
	where $F \subseteq E$ is a violating edge set that separates $u$ 
	from $v$ in $H$, and $\exists \calK_j$ with $uv \in \calK_j$ and
	$(F, \calK_j) \in \Omega_{\ell}$,
	\item \emph{Elements:} The set of elements is $E \setminus E(H)$.
	An edge $e$ \emph{hits} a tuple $(F, uv)$ if $u$ and $v$ are connected in 
	$(H  \cup \{e\} \cup P_T(e)) \setminus F$. The cost of this element 
	is $c(\{e\} \cup P_T(e))$. 
\end{itemize}

Note that the number of violating edge sets $F$ is at most
$\binom {m}{\ell-1} \leq n^{2\ell}$, and the number of elements is at
most $n^2$. There are polynomial-time algorithms that given a Hitting
Set instance (which is the same as Set Cover in the dual set system)
with $\alpha$ sets and $\beta$ elements outputs an $O(\log \alpha)$
approximate feasible solution
\cite{Vazirani-approx,WilliamsonS-approx}.  Thus, in time
$n^{O(\ell)}$, we can use such an algorithm to obtain an
$O(\log(n^{2\ell})) = O(\ell \log n)$-approximation to the preceding
instance of Hitting Set that we set up. Let $E'$ be such a solution
and let $H' = \cup_{e \in E'} (\{e\} \cup P_T(e))$. We return
$H_P \cup H'$. 

\subsection{Cost and Correctness Proofs}

We begin by proving correctness of the above algorithm. 

\begin{lemma}
\label{lem:sndp_correctness}
	$H \cup H'$ satisfies all scenarios in $\Omega_{\ell}$.
\end{lemma}
\begin{proof}
	Suppose for the sake of contradiction that there exists some $(F, \calK_j) \in
	\Omega_{\ell}$ not satisfied by $H \cup H'$. Let $(u,v) \in \calK_j$ be a 
	terminal pair disconnected in $(H \cup H') \setminus F$. By construction,
	$(F, uv)$ must have been a set in the Hitting Set instance above. 
	By Lemma \ref{lem:cutcover}, there exists some edge $e$ such that 
	$(\{e\} \cup P_T(e)) \setminus F$ connects the component of $H \setminus F$ 
	containing $u$ to the component of $H \setminus F$ containing $v$. In 
	particular, this means that $u$ and $v$ are connected in 
	$(H \cup \{e\} \cup P_T(e)) \setminus F$. Thus
	there exists an element in the Hitting Set instance that hits $(F, uv)$. 
	Since $E'$ is a feasible solution, it must contain an element that hits $(F, uv)$,
	so $H'$ must contain a fundamental cycle covering $(F, uv)$.
\end{proof}

Next, we bound the expected cost of this algorithm.

\begin{lemma}
\label{lem:sndp_h0_cost}
  $\E[c(H_P)] \leq O(\log n \log \log n) \opt$.
\end{lemma}
\begin{proof}
	Fix an optimal solution $E^*$ to the given instance of Bulk-SNDP. Let 
	$H_P^* = \cup_{e \in E^*} P_T(e)$ be the tree paths connecting the endpoints
	of all edges in $E^*$. By Lemma \ref{lem:spanningtree_embeddings}, 
	$\E[c(H_P^*)] \leq \sum_{e \in E^*} \E[c(P_T(e))] \leq 
	\sum_{e \in E^*} O(\log n \log \log n) c(e) = O(\log n \log \log n) \opt$.
	Finally, note that for every terminal pair $(u,v)$, $H_P^*$ must contain 
	the unique tree path between $u$ and $v$, since $E^*$ connects all 
	terminal pairs so $H_P^*$ does as well. Therefore, $H_P \subseteq H_P^*$, 
	so $\E[c(H_P)] \leq \E[c(H_P^*)] \leq O(\log n \log \log n) \opt$.  
\end{proof}

\begin{lemma}
\label{lem:sndp_cost}
  $\E[c(H')] \leq O(\ell \log^2n \log \log n) \opt$.
\end{lemma}
\begin{proof}
	Let $E^*$ be an optimal solution to the augmentation problem. We will
	apply the cut-cover lemma (Lemma \ref{lem:cutcover}) to show that $E^*$
	is a feasible solution to the Hitting Set instance. First, note that $T$
	is a spanning tree and
	$H$ contains all tree paths connecting terminal pairs, since $H$ contains $H_P$. 
	Next, note that for all $(F,uv)$ corresponding to sets in the Hitting Set
	instance, $E^*$ connects $u$ and $v$ in $E^* \setminus F$, since $E^*$ is a 
	feasible solution to the augmentation problem. Finally, note that any 
	violated set $F$ must be minimal, since $H$ satisfies all scenarios in 
	$\Omega_{\ell-1}$, which includes all strict subsets of $F$. Thus, by 
	Lemma \ref{lem:cutcover}, for each $(F, uv)$ corresponding to a set in the 
	Hitting Set instance, there exists some $e \in E^*$ that hits $(F, uv)$. 
	Therefore, $E^*$ is a feasible solution to the Hitting Set instance.
	Since $E'$ is an $O(\ell \log n)$-approximation to the Hitting Set problem, 
	$c(E') \leq O(\ell \log n) c(E^*) = O(\ell \log n) \opt$.
	By Lemma \ref{lem:spanningtree_embeddings}, 
	\[\E[c(H')] = \sum_{e \in E'} (c(e)) + c(P_T(e)) \leq O(\log n \log\log n) c(E')
	\leq O(\ell \log^2 n \log\log n) \opt.\] 
\end{proof}

Thus the total expected cost of the augmentation algorithm is at most 
$\E[c(H_P)] + \E[c(H')] \leq O(\ell \log^2 n \log \log n) \opt$.
Summing over all $\ell \in [k]$, we get a total cost of 
$O(k^2 \log^2n \log \log n) \opt$,
concluding the proof of Theorem \ref{thm:bulk-sndp}. Corollary \ref{cor:rsndp}
for RSNDP follows immediately. For Flex-SNDP, we can start with a 2-approximation 
for $(p,0)$-Flex-SNDP, since this is equivalent to EC-SNDP where all terminals 
have connectivity requirement $p$, and augment from $(p, q-1)$ to $(p, q)$. 
Thus the total cost is $\sum_{\ell \in \{p, \dots, p+q\}} 
O(\ell \log^2 n \log \log n) \opt \leq
O(q(p+q) \log^2n \log \log n) \opt$, concluding 
the proof of Corollary \ref{cor:flex-sndp}.

\begin{remark}
  In the preceding proof of Corollary \ref{cor:flex-sndp}, we assumed
  that all demand pairs $s_i, t_i$ have the same connectivity
  requirement $(p,q)$.  We claim that the same approximation bound can
  be achieved in the setting where the requirement for pairs can be
  different with the condition that $p_i \le p$ and $q_i \le q$ for all $i$.  We start by
  solving an EC-SNDP instance in which the terminal pair $s_i, t_i$
  has connectivity requirement $p_i$.  We then apply the cut cover
  idea approach as described, in $q$ stages. In stage $j$, a pair
  $s_i,t_i$ is unsatisfied iff $q_i \ge j$; we augment connectivity
  for each unsatisfied pair from $(p_i, j-1)$ to $(p_i, j)$. One can
  see that the correctness and cost analysis is valid in this more general setting.
\end{remark}

\section{LP-Based Approximation for Bulk-SNDP}
\label{sec:sndp-racke}

In this section, we prove Theorem~\ref{thm:bulk-group} and the
resulting Corollaries~\ref{cor:flex-group} and \ref{cor:rgroup}. We
obtain poly-logarithmic approximations for Bulk-SNDP, Flex-SNDP, and
RSNDP with respect to their LP relaxations (see Section
\ref{sec:prelim}), and obtain approximations for the group
connectivity variants as well. We describe the algorithm and proofs
for the group connectivity setting since it generalizes SNDP.  Recall
that we are given as input a graph $G: (V, E)$ with cost function
$c: E \to \R_{\geq 0}$ and a set of scenarios
$\Omega = \{(F_j, \calK_j): j \in [m]\}$, where each $\calK_j$
consists of set pairs $(S_i, T_i)$, $S_i, T_i \subseteq V$.  Assume we
have a partial solution $H$ satisfying all scenarios in
$\Omega_{\ell-1}$. We augment $H$ to satisfy scenarios in
$\Omega_{\ell}$.

As in Section \ref{sec:sndp_cutcover}, our high level approach is to cover 
violated edge cuts rather than vertex cuts. We do so by adapting an LP 
rounding technique given by \cite{ChenLLZ22} for Group SNDP. We begin with some 
background on \racke's capacity-based probabilistic tree embeddings and 
Tree Rounding algorithms for Group Steiner tree.

\subsection{\racke Tree Embeddings}
\label{subsec:sndp_racke}
The results in this section use \racke's capacity-based probabilistic
tree embeddings. We borrow the notation from~\cite{ChenLLZ22}. 
Given $G = (V, E)$ with capacity $x: E \to \R^+$
on the edges, a capacitated tree embedding of $G$ is a tree $\calT$,
along with two mapping functions $\calM_1: V(\calT) \rightarrow V(G)$
and $\calM_2: E(\calT) \rightarrow 2^{E(G)}$ that satisfy some
conditions. $\calM_1$ maps each vertex in $\calT$ to a vertex in $G$,
and has the additional property that it gives a one-to-one mapping
between the leaves of $\calT$ and the vertices of $G$.  $\calM_2$ maps
each edge $(a,b) \in E(\calT)$ to a path in $G$ between $\calM_1(a)$
and $\calM_1(b)$.  For notational convenience we view the two mappings
as a combined mapping $\calM$. For a vertex $u \in V(G)$ we use
$\calM^{-1}(u)$ to denote the leaf in $\calT$ that is mapped to $u$ by
$\calM_1$. For an edge $e \in E(G)$ we use $\calM^{-1}(e) = \{f \in
E(\calT)\mid e \in \calM_2(f)\}$. It is sometimes convenient to
view a subset $S \subseteq V(G)$ both as vertices in $G$ and
also corresponding leaves of $\calT$.

The mapping $\calM$ induces a capacity function $y: E(\calT)
\rightarrow \mathbb{R}_+$ as follows. Consider $f=(a,b) \in
E(\calT)$. $\calT - f$ induces a partition $(A,B)$ of $V(T)$ which in
turn induces a partition/cut $(A',B')$ of $V(G)$ via the mapping
$\calM$: $A'$ is the set of vertices in $G$ that correspond to the
leaves in $A$ and similarly $B'$. We then set $y(f) = \sum_{e \in
  \delta(A')} x(e)$, in other words $y(f)$ is the capacity of cut
$(A',B')$ in $G$. The mapping also induces loads on the edges of
$G$. For each edge $e \in G$, we let $\load(e) = \sum_{f \in E(\calT):
  e \in M(f)} y(f)$.  The relative load or \emph{congestion} of $e$ is
$\rload(e) = \load(e)/x(e)$.  The congestion of $G$ with respect to a
tree embedding $(\calT, \calM)$ is defined as $\max_{e \in E(G)}
\rload(e)$. Given a probabilistic distribution $\calD$ on trees
embeddings of $(G,x)$ we let
$\beta_{\calD} = \max_{e \in E(G)}  \E_{(\calT,\calM) \sim \calD} \rload(e)$
denote the maximum expected congestion.
\racke showed the following fundamental result on probabilistic embeddings 
of a capacitated graph into trees.

\begin{theorem}[\cite{Racke08}]
\label{thm:racke}
  Given a graph $G$ and $x: E(G) \to \R^+$, there exists a probability
  distribution $\calD$ on tree embeddings such that $\beta_{\calD} =
  O(\log |V(G)|)$.  All trees in the support of $\calD$ have height at
  most $O(\log(nC))$, where $C$ is the ratio of the largest to smallest
  capacity in $x$. Moreover, there is randomized polynomial-time
  algorithm that can sample a tree from the distribution $\calD$.
\end{theorem}

In the rest of the paper we use $\beta$ to denote the guarantee
provided by the preceding theorem where $\beta =
O(\log n)$ for a graph on $n$ nodes. In order to use these probabilistic 
embeddings to route flow, we need the following corollary, where we use
$\maxflow_H^z(A,B)$ to denote the maxflow between two disjoint vertex
subsets $A,B$ in a capacitated graph $H$ with capacities given by $z:
E(H) \rightarrow \mathbb{R}_+$.

\begin{corollary}
\label{cor:racketreeflow}
  Let $\calD$ be the distribution guaranteed in Theorem~\ref{thm:racke}.
  Let $A, B \in V(G)$ be two disjoint sets. Then \\
  (i) for any tree
  $(\calT,\calM)$ in $\calD$, $\maxflow_G^x(A, B) \leq 
  \maxflow_\calT^y(\calM^{-1}(A), \calM^{-1}(B))$ and \\(ii) $\frac 1 \beta
  \E_{(\calT,\calM) \sim \calD}[\maxflow_\calT^y(\calM^{-1}(A),
  \calM^-(B))] \leq \maxflow_G^x(A, B)$.
\end{corollary}

\subsection{Oblivious Tree Rounding}
\label{subsec:tree_rounding}

As described in Section \ref{subsec:related_work}, Chalermsook, Grandoni
and Laekhanukit \cite{CGL15} studied the Survivable Set Connectivity 
problem. They describe a randomized oblivious algorithm based 
on group Steiner tree rounding from \cite{GargKR98}. This is
useful since the sets to be connected during the course of their
algorithm are implicitly generated. We encapsulate their result in the
following lemma. The tree rounding algorithm in \cite{CGL15,ChenLLZ22}
is phrased slightly differently since they combine aspects of group
Steiner rounding and the congestion mapping that comes from \racke
trees. We separate these two explicitly to make the idea more
transparent. We refer to the algorithm from the lemma below as TreeRounding.

\begin{lemma}[\cite{CGL15,ChenLLZ22}]
\label{lem:setconnectivity-tree-rounding}
  Consider an instance of Set Connectivity on an $n$-node tree $T=(V,E)$
  with height $h$ and let $x: E \rightarrow [0,1]$. Suppose $A, B
  \subseteq V$ are disjoint sets and suppose $K \subseteq E$ such that
  $x$ restricted to $K$ supports a flow of $f \le 1$ between $A$ and
  $B$. There is a randomized algorithm that is oblivious to $A, B, K$
  (hence depends only on $x$ and value $f$) that outputs a subset $E' \subseteq E$
  such that (i) The probability that $E' \cap K$ connects $A$ to $B$ is
  at least a fixed constant $\phi$ and (ii) For any edge $e \in E$, the
  probability that $e \in E'$ is $\min\{1,O(\frac{1}{f} h \log^2 n) x(e)\}$.
\end{lemma}

\subsection{Rounding Algorithm for the Augmentation Problem}
\label{subsec:sndp_algo}

\newcommand{\LG}{\textnormal{LARGE}}
\newcommand{\SM}{\textnormal{SMALL}}
\newcommand{\lpopt}{\text{OPT}_{\text{LP}}} We adapt the algorithm and
analysis in~\cite{ChenLLZ22} to Group Bulk-SNDP. Let $\beta$ be the
expected congestion given by Theorem~\ref{thm:racke}.
We start by obtaining a solution $\{x_e\}_{e \in E \setminus H}$ for
the LP relaxation (described in Section \ref{sec:prelim}). Let $E' = E \setminus H$.  
We define $\LG = \{e \in E': x_e \geq \frac 1 {4\ell\beta}\}$, and
$\textnormal{SMALL} = \{e \in E': x_e < \frac 1 {4\ell\beta}\}$.  The
LP has paid for each $e \in \LG$ a cost of at least $c(e)/(4\ell\beta)$, hence
adding all of them to $H$ will cost $O(\ell\beta \cdot \lpopt)$.  If
$\LG \cup H$ is a feasible solution to the augmentation problem, then
we are done since we obtain a solution of cost $O(\ell\log n \cdot
\lpopt)$.  Thus, the interesting case is when $\LG \cup H$ is \emph{not} a
feasible solution.  

Following \cite{ChenLLZ22} we employ a \racke tree based rounding.
A crucial step is to set up a capacitated graph appropriately. We can
assume, with a negligible increase in the fractional cost, that for
each edge $e \in E'$, $x(e) = 0$ or $x(e) \ge
\frac{1}{n^3}$; this can be ensured by rounding down to $0$ the
fractional value of any edge with very small value, and compensating
for this loss by scaling up the fractional value of the other edges by
a factor of $(1+1/n)$. It is easy to check that the new solution
satisfies the cut covering constraints, and we have only increased the
cost of the fractional solution by a $(1+1/n)$-factor. In the
subsequent steps we can ignore edges with $x_e = 0$ and assume that
there are no such edges.

Consider the original graph $G=(V,E)$ where we set a capacity for each
$e \in E$ as follows. If $e \in \LG \cup H$ we set $\tilde x_e = \frac
1 {4\ell\beta}$. Otherwise we set $\tilde x_e = x_e$.  Since the
ratio of the largest to smallest capacity is $O(n^3)$, the height of
any \racke tree for $G$ with capacities $\tilde x$ is at most
$O(\log n)$.  Then, we repeatedly sample \racke trees. For each
tree, we sample edges by the rounding algorithm given by Chalermsook
et al. in \cite{CGL15} (see Section~\ref{subsec:tree_rounding} for details). A
formal description of the algorithm is provided below where $t'$ and
$t$ are two parameters that control the number of trees sampled and
the number of times we run the tree rounding algorithm in each sampled
tree.  We will analyze the algorithm by setting both $t$ and $t'$ to 
$\Theta(\ell\log n)$.

\begin{algorithm}[H]
\caption{Approximating the Bulk-SNDP Augmentation Problem from $\ell-1$ to $\ell$}
\label{augmentation_algo}
\begin{algorithmic}
    \State $H \gets$ partial solution satisfying scenarios in $\Omega_{\ell}$
    \State $\{x\}_{e \in E} \gets$ fractional solution to the LP
    \State $\textnormal{LARGE} \gets \{e \in E': x_e \geq \frac 1 {4\ell\beta}\}$
    \State $\textnormal{SMALL} \gets \{e \in E': x_e < \frac 1 {4\ell\beta}\}$
    \State $H \gets H \cup \textnormal{LARGE}$
    \If{$H$ is a feasible solution satisfying scenarios in $\Omega_{\ell+1}$}
        \Return $H$
    \Else 
        \State $\tilde x_e \gets \begin{cases}
            \frac 1 {4\ell\beta} & e \in H \\
            x_e & \text{otherwise}
            \end{cases}$
    \EndIf
    \State $\calD \gets $ \racke tree distribution for $(G, \tilde x)$
    \For{$i = 1, \dots t'$}
        \State Sample a tree $(\calT, \calM, y) \sim \calD$
        \For{$j = 1, \dots, t$}
            \State $H' \gets $ output of oblivious TreeRounding algorithm on $(G, \calT)$
            \State $H \gets H \cup \calM(H')$
        \EndFor
    \EndFor \\
    \Return H
\end{algorithmic}
\end{algorithm}

\subsection{Analysis}
\label{subsec:sndp_analysis}
For the remainder of this analysis, we denote as $H$ the partial solution after 
buying edges in \LG. We will assume, following earlier discussion, that $H$ 
does not satisfy all requirements specified by scenarios in $\Omega_{\ell}$. 
This implies that there must be some $F$ such that $|F| \leq \ell$, 
$F \subseteq F_j$ for some $j \in [m]$, and $\exists (S_i, T_i) \in \calK_j$ such 
that $S_i$ and $T_i$ are disconnected in $(V, H \setminus F)$. Since $H$ satisfies 
all scenarios in $\Omega_{\ell-1}$, it must be the case that $F$ has exactly 
$\ell$ edges. We call such an $F$ a violating edge set.

As discussed in Section \ref{sec:sndp_cutcover}, it suffices to show that
for each violated edge set $F$ that disconnects some terminal pair $S_i, T_i$,
the algorithm outputs a set $H'$ that contains a path from $S_i$ to $T_i$ 
in $(H \cup H') \setminus F$ with high probability.
We observe that 
the algorithm is oblivious to $F$. Thus, if we obtain a high probability bound 
for a fixed $F$, since there are at most $O(n^{2\ell})$ violating edge sets, we can 
use the union bound to argue that $H'$ covers \emph{all} 
violating edge sets.  For the remainder of this section, until we do the final 
cost analysis, we work with a fixed violating edge set $F$. For ease of notation, 
we let $\calK_F = \bigcup_{j \in [m]: F \subseteq F_j} \calK_j$ be the set of 
terminal pairs that need to be connected in $(H \cup H') \setminus F$.

Consider a tree $(\calT,\calM, y)$ in the \racke distribution for the
graph $G$ with capacities $\tilde x$. We let $\calM^{-1}(F)$ denote
the set of all tree edges corresponding to edges in $F$,
i.e. $\calM^{-1}(F) = \cup_{e \in F} \calM^{-1}(e)$.  We call $(\calT,
\calM, y)$ \emph{good} with respect to $F$ if $y(\calM^{-1}(F)) \leq
\frac 1 2$; equivalently, $F$ blocks a flow of at most $\frac 1 2$ in
$\calT$.

\begin{lemma}
\label{lemma:sndp_good_tree}
  For a violating edge set $F$, a randomly sampled \racke tree 
  $(\calT, \calM, y)$ is good with respect to $F$ with probability at least 
  $\frac 1 2$.
\end{lemma}
\begin{proof}
  For each $e \in F$, $\tilde x_e =\frac 1 {4\ell\beta}$. Since the expected 
  congestion of each edge is at most $\beta$, $\E[\load(e)] \leq 
  \beta \tilde x_e \le \frac 1 {4\ell}$ for each $e \in F$. Note that 
  $y(\calM^{-1}(F)) = \sum_{e \in F} \load(e)$, hence by linearity of expectation, 
  $\E[y(\calM^{-1}(F)] = \sum_{e \in F} \E[\load(e)] \leq |F|\frac 1 {4\ell} 
  = \frac 1 4$. Applying Markov's inequality to $y(\calM^{-1}(F))$ proves 
  the lemma.
\end{proof}

Given the preceding lemma, a natural approach is to sample a good tree 
$\calT$ and hope that $\calT \setminus M^{-1}(F)$ still has good flow between 
each terminal pair. However, since we rounded down all edges in $\LG \cup H$, 
it is possible that $\calM^{-1}(F)$ contains an edge whose removal would 
disconnect a terminal pair in $\calT$, even if $\calT$ is good. See 
\cite{ChenLLZ22} for a more detailed discussion and example.

We note that our goal is to find a set of
edges $H' \subseteq E$ such that each terminal pair in $\calK_F$ has a path in
$(H' \cup H) \setminus F$; these paths must exist in the original
graph, even if they do not exist in the tree. Therefore, instead of
looking directly at paths in $\calT$, we focus
on obtaining paths through components that are already connected in
$(V(G), H \setminus F)$. The rest of the argument is to show that
sufficiently many iterations of TreeRounding on any good tree $\calT$
for $F$ will yield a feasible set $H'$ for $F$.

\subsection{Shattered Components, Set Connectivity and Rounding}
Let $\Q_F$ be the set of connected components in the subgraph induced by 
$H \setminus F$. We use vertex subsets to denote components.
Let $\calT$ be a good tree for $F$. We say that a connected component 
$Q \in \Q_F$ is \emph{shattered} if it is disconnected in $\calT \setminus 
\calM^{-1}(F)$, else we call it \emph{intact}. For any $S \subseteq V$, 
let $Q_{S}$ be the union of components of $Q_F$ that have a nonempty
intersection with $S$.
Note that for some $(S_i, T_i) \in \calK_F$, $Q_{S_i}$ may be the same as 
$Q_{T_i}$, but if $F$ is a violating edge set then there 
is at least one pair such that $Q_{S_i} \cap Q_{T_i} = \emptyset$. 
Now, we define a Set Connectivity instance that is induced by $F$ and $\calT$. 
Consider two disjoint vertex subsets $A,B \subset V$.
We say that $(A,B)$ partitions the set of shattered components if each shattered 
component $Q$ is fully contained in $A$ or fully contained in $B$. 
Formally let 
$$Z_F = \{(A \cup Q_{S_i}, B \cup Q_{T_i}): (A, B)
\text{ partitions shattered components}, (S_i,T_i) \in \calK_F, Q_{S_i} \cap Q_{T_i} = \emptyset\}.$$ 
In other words, $Z_F$ is set of all partitions of shattered components that
separate some pair $(S_i,T_i) \in \calK_F$.  Since the leaves of $\calT$ are in 
one to one correspondence with $V(G)$ we can view $Z_F$ as inducing a 
Set Connectivity instance in $\calT$; technically we need to consider the pairs 
$\{(\calM^{-1}(A),\calM^{-1}(B)) \mid (A,B) \in Z_F\}$; however, for 
simplicity we conflate the leaves of $\calT$ with $V(G)$.  We claim that it 
suffices to find a feasible solution that connects the pairs defined by 
$Z_F$ in the tree $\calT$.

\begin{lemma}
\label{lemma:sndp_shattered_suffices}
  Let $E' \subseteq \calT \setminus \calM^{-1}(F)$. Suppose there exists a path in 
  $E' \subseteq \calT \setminus \calM^{-1}(F)$ connecting $A$ to $B$ for all 
  $(A, B) \in Z_F$. Then, there is an $S_i$-$T_i$ path for each $(S_i,T_i) \in \calK_F$ 
  in $(\calM(E') \cup H) \setminus F$. 
\end{lemma}
\begin{proof}
  Let $E' \subseteq \calT \setminus \calM^{-1}(F)$ such that there is a path 
  from $A$ to $B$ in $E'$ for each $(A, B) \in Z_F$. Assume for the sake of 
  contradiction that $\exists (S_i,T_i) \in \calK_F$ such that $S_i$ and $T_i$ are 
  disconnected in $(\calM(E') \cup H) \setminus F$. 
  In particular, this means that $Q_{S_i}$ and $Q_{T_i}$ are disjoint.
  Then, there must be some cut 
  $X$ such that $\delta_{(\calM(E') \cup H) \setminus F}(X) = \emptyset$ and 
  $S_i \subseteq X$ and $T_i \subseteq \overline X$.
  
  We observe that no component $Q \in \Q_F$ can cross $X$ since each $Q$
  is connected in $H\setminus F$. 
  Let $A = Q_{S_i} \cup \{Q \in \Q_F: Q \text{ is shattered }, Q \subseteq X\}$,
  and $B = Q_{T_i} \cup \{Q \in \Q_F: Q
  \text{ is shattered }, Q \subseteq \overline X\}$. Clearly, $A
  \subseteq X$, $B \subseteq \overline X$. Furthermore, $(A, B) \in
  Z_F$. By
  assumption, there is a path $P$ in $E'$ between $A$ and $B$. Since $E'
  \cap \calM^{-1}(F) = \emptyset$, $\calM(E')$ cannot contain any edges
  in $F$. Therefore, $\calM(P)$ contains a path that crosses $X$ which
  implies that $|\delta_{\calM(E')}(X)| = |\delta_{\calM(E') \setminus
  F}(X)| \geq 1$, contradicting the assumption on $X$.
\end{proof}

\mypara{Routing flow:} We now argue that $(\calT, \calM, y)$ routes sufficient 
flow for each pair in $Z_F$ without using the edges in $\calM^{-1}(F)$; 
in other words $y$ is fractional solution (modulo a scaling factor) to the 
Set Connectivity instance $Z_F$ in the graph/forest 
$\calT \setminus \calM^{-1}(F)$. We can then appeal to TreeRounding lemma to 
argue that it will connect the pairs in $Z_F$ without using any edges in $F$. 

\begin{lemma}
\label{lemma:sndp_flowforeachpair}
  Let $(A, B) \in Z_F$. Let
  $X \subset V_{\calT}$ such that $A \subseteq X$ and $B \subseteq V_{\calT} 
  \setminus X$. Then $y(\delta_{\calT \setminus \calM^{-1}(F)}(X)) 
  \geq \frac 1 {4\ell\beta}$.
\end{lemma}
\begin{proof}
  Let $X \subseteq V_\calT$ be a vertex set that separates $A$ from $B$, 
  i.e. $A \subseteq X$ and $B \subseteq V_\calT \setminus X$. First,
  suppose there exists a component $Q \in \Q_F$ such that $Q$ crosses
  $X$, i.e. $X \cap Q \neq \emptyset$ and $\overline X \cap Q \neq
  \emptyset$. Since $(A, B)$ partitions the set of shattered components,
  $Q$ must be intact in $\calT$. Let $u$ be a leaf in $Q \cap X$ and $v$
  be a leaf in $Q \cap \overline X$.  Since $Q$ is intact in $\calT$, $u$ and $v$ 
  are connected in $\calT \setminus \calM^{-1}(F)$ by some path $P$. Since 
  $u \in X$, $v \in \overline X$, $P$ crosses $X$. Let 
  $e \in P \cap \delta_{\calT \setminus \calM^{-1}(F)}(X)$. 
  It remains to show that
  $y(e) \ge \frac 1 {4\ell\beta}$. This follow from properties of the
  \racke tree. Since $u$ and $v$ are connected in $G'$ with a path using only 
  edges in $H$ each of which has a capacity of $\frac 1
  {4\ell\beta}$, $u$-$v$ maxflow in $G'$ is at least $\frac 1
  {4\ell\beta}$. From Corollary~\ref{cor:racketreeflow},
  for any tree $\calT$, the $u$-$v$
  maxflow in $\calT$ with capacities $y$ must be at least $\frac 1
  {4\ell\beta}$. This in particular implies that $y(e) \ge \frac 1
  {4\ell\beta}$ for every edge $e \in P$.
  
  We can now restrict attention to the case that no connected component
  of $\Q_F$ crosses $X$. Let $X'$ be the set of leaves in $X$ and
  consider the cut $(X',V\setminus X')$ in $G$. It follows that
  $(X',V \setminus X')$ partitions the connected components in $\Q_F$ and
  $\delta_{H \setminus F}(X') = \emptyset$. Since $(A,B) \in Z_F$ there is a pair
  $(S_i, T_i) \in \calK_F$ such that $Q_{S_i} \subseteq S'$ and $Q_{T_i} \subseteq V\setminus
  S'$. Thus $(S',V\setminus S')$ is a violated cut with $F$ as its
  witness. Recall that there exists some $j \in [m]$ such that 
  $F \subseteq F_j$ and $(S_i,T_i) \subseteq \calK_j$. Since $\{x_e\}_{e \in E}$ is a 
  feasible solution to the LP it follows that 
  $x(\delta_{E \setminus F}(X')) \ge x(\delta_{E \setminus F_j}(X')) \ge 1$. 
  Since $\delta_{H \setminus F}(X') = \emptyset$, the edges in 
  $\delta_{E \setminus F}(X')$ are in $\SM$. Therefore, 
  $x(\delta_{E \setminus F}(X')) = \tilde x(\delta_{E \setminus F}(X')) \ge 1$.
  
  The \racke tree property guarantees that $y(\delta_{\calT}(X)) 
  \ge \tilde x(\delta_{G}(X')) \ge 1$ (via Corollary~\ref{cor:racketreeflow}). 
  We note that 
  $$y(\delta_{\calT \setminus \calM^{-1}(F)}(X)) \ge y(\delta_{\calT}(X)) 
  - y(\calM^-(F)) \ge 1 - 1/2 \ge 1/2,$$
  where we used the fact that $y(\calM^-(F)) \le 1/2$ since $\calT$ is good for $F$.
  Thus in both cases we verify the desired bound.
\end{proof}

\mypara{Bounding $Z_F$:} A second crucial property is a bound on $|Z_F|$,
the number of pairs in the Set Connectivity instance induced by $F$ and a 
good tree $\calT$ for $F$. 

\begin{lemma}
\label{lemma:sndp_boundnumpairs}
  For a good tree $\calT$, $|Z_F| \leq 2^{2\ell\beta} |\calK_F|$.
\end{lemma}
\begin{proof}
  Let $h$ be the number of shattered components and let them be
  $Q_1,\ldots,Q_h$. For each $Q_i$ pick a pair of vertices $u_i,v_i$
  that are in separate components of $\calT - \calM^{-1}(F)$. Let 
  $A = \{u_1,\ldots,u_h\}$ and $B = \{v_1,v_2,\ldots,v_h\}$. Since the
  paths in $G$ connecting $u_i,v_i$ are each in different connected components 
  of $H\setminus F$, it follows that the $(A,B)$-maxflow in $H \setminus F$
  is at least $h$. In the graph $G'$ obtained by scaling down the
  capacity of edges of $H$, the maxflow is at least $\frac h
  {4\ell\beta}$ which implies that it is at least this quantity in
  $\calT$. Since $\calT$ is good, the total decrease of flow can be at
  most $y(\calM^{-1}(F)) \le \frac 1 2$. By construction there is no
  flow between $A$ and $B$ in $\calT - \calM^{-1}(F)$ which implies that
  $\frac h {4\ell\beta} \le 1/2 \Rightarrow h \leq 2\ell\beta$.
  Each pair in $Z_F$ corresponds to a subset of shattered components and
  a demand pair $(S_i,T_i) \in \calK_F$, and hence $|Z_F|\leq 2^h |\calK_F| \le
  2^{2\ell\beta} |\calK_F|$.
\end{proof}

\subsection{Correctness and Cost}
The following two lemmas show that by taking a union bound over all violating 
edge sets $F$ and applying the Tree Rounding lemma 
\ref{lem:setconnectivity-tree-rounding}, one can show that the algorithm 
outputs a feasible augmentation solution with probability at least $\frac 1 2$. 

\begin{lemma}
\label{claim:successprobforgoodtree}
  Suppose $\calT$ is good for a violating edge set $F$. Then after 
  $t = O(\log r + \ell \log n)$
  rounds of TreeRounding with flow parameter $\frac 1 {4\ell\beta}$,
  the probability that $H'$ is \emph{not} a feasible augmentation for
  $F$ is at most $(1-\phi)^t |Z_F| \le 1/4$.
\end{lemma}
\begin{proof}
  Suppose $\calT$ is good for $F$. Let $(A,B) \in Z_F$. From
  Lemma~\ref{lemma:sndp_flowforeachpair} the flow for $(A,B)$ in $\calT -
  \calM^{-1}(F)$ is at least $\frac 1 {4\ell\beta}$. From
  Lemma~\ref{lem:setconnectivity-tree-rounding}, with probability at
  least $\phi$, the pair $(A,B)$ is connected via a path in $\calT -
  \calM^{-1}(F)$. If all pairs are connected, then via
  Lemma~\ref{lemma:sndp_shattered_suffices}, $H'$ is a feasible
  augmentation for $F$. Thus, $H'$ is not a feasible augmentation if
  for some $(A,B) \in Z_F$ the TreeRounding does not succeed after $t$
  rounds. The probability of this, via the union bound over the pairs
  in $Z_F$, is at most $(1-\phi)^t |Z_F|$. From
  Lemma~\ref{lemma:sndp_boundnumpairs}, $|Z_F| \le 2^{2\ell\beta}|\calK_F|$. 
  Note that $|\calK_F| \leq r$. Consider $t
  = \frac 1 \phi \log(4|\calK_F| \cdot 2^{2\ell\beta}) 
  = O(\log r + \ell \log n)$, 
  since $\beta = O(\log n)$. Then, $(1-\phi)^t |Z_F| 
  = 2^{2\ell\beta} |\calK_F| (1 -\phi)^t 
  \leq 2^{2\ell\beta} |\calK_F| e^{-\phi t} \leq \frac 1 4$.
\end{proof}

\begin{lemma}
\label{lem:correctness}
  The algorithm outputs a 
  feasible augmentation to the given instance with probability at least
  $\frac 1 2$.
\end{lemma}
\begin{proof}
  For a fixed $F$ the probability that a sampled tree is good is at
  least $1/2$.  By Claim~\ref{claim:successprobforgoodtree}, conditioned
  on the sampled tree being good for $F$, $t$ iterations of TreeRounding
  fail to augment $F$ with probability at most $1/4$. Thus the probability that
  all $t'$ iterations of sampling trees fail is $(1-3/8)^{t'}$. There
  are at most $n^{2\ell}$ violating edge sets $F$. Consider $t' = \frac 8 3
  \log (2n^{2\ell}) = O(\ell\log n)$. By applying the union bound over
  all violating edge sets $F$, the probability of the algorithm failing
  is at most $n^{2\ell}(1 - 3/8)^{t'} \leq n^{2\ell}e^{-3t'/8} \leq
  \frac 1 2$. Therefore, the output of the algorithm is a feasible
  augmentation for all violating edge sets with probability at least
  $\frac 1 2$.
\end{proof}

Now we analyze the expected cost of the edges output by the algorithm for augmentation
with respect to $\lpopt$, the cost of the fractional solution.

\begin{lemma}
\label{lem:costanalysis}
  The total expected cost of the algorithm is 
  $O((\log r + \ell\log n)\ell^2 \log^6 n) \cdot \lpopt$.
\end{lemma}
\begin{proof}
  Fix an edge $e \in \SM$ with fractional value $x_e$. Consider one
  outer iteration of the algorithm in which it picks a random tree
  $\calT$ from the \racke tree distribution and then runs $t$ iterations
  of TreeRounding with flow parameter $\alpha = \frac 1
  {4\ell\beta}$. Via Lemma~\ref{lem:setconnectivity-tree-rounding}, the
  probability of an edge $f \in \calT$ being chosen is at most
  $O(\frac{1}{\alpha} h \log^2n) y(f)$, where $h$ here denotes the height of 
  the tree $\calT$. Thus the expected cost for $e$
  for one round of TreeRounding is $O(\frac{1}{\alpha} h \log^2n)
  \sum_{f \in \calM^{-1}(e)} y(f) = O(\frac{1}{\alpha} h \log^2n)
  \load(e)$. By the \racke distribution property, $\E_{\calT} [\load(e)]
  \le \beta x_e$. By linearity of expectation, since there are a total
  of $t \cdot t'$ iterations of TreeRounding, the total expected cost is
  at most $(t \cdot t')\cdot O(\frac{1}{\alpha} h \log^2 n \beta) \sum_{e \in
    E} c(e) x_e$. By the analysis in Section~\ref{subsec:sndp_algo}, $h = O(\log
  n)$, and $\beta = O(\log n)$. Substituting in the values of $t$ and $t'$ stated
  in Lemmas \ref{claim:successprobforgoodtree} and \ref{lem:correctness}, 
  the total expected cost is at most 
  $O((\log r + \ell\log n)\ell^2 \log^6 n) \cdot \lpopt$.
\end{proof}

Combining the correctness and cost analysis we obtain the following.
\begin{lemma}
\label{lem:sndp_augmentation_result}
  There is a randomized $O((\log r + \ell\log n)\ell^2 \log^6 n)$-approximation algorithm for the 
  Bulk-SNDP Augmentation problem from $\ell-1$ to $\ell$. The algorithm runs 
  in time polynomial in $n$ and $\alpha$, where $\alpha$ is the amount of time 
  it takes to solve the LP. 
\end{lemma}

To prove Theorem~\ref{thm:bulk-group}, we start with a solution from $\ell = 0$ 
and iteratively solve $k$ augmentation problems. Since the LP for Group Bulk-SNDP 
can be solved in polynomial time (see Section~\ref{sec:prelim}), we obtain a 
polynomial time $O((\log r + k\log n)k^3\log^6 n)$-approximation algorithm.

For flex-connectivity, 
we reduce to Bulk-SNDP with maximum width $p+q$.
Recall from Section~\ref{sec:prelim} that the LP 
can be solved in $n^{O(q)}$ time, giving us Corollary~\ref{cor:flex-group}.

Finally, for Relative SNDP, there is an LP relaxation described 
in~\cite{DinitzKK22} that can be solved in polynomial time, even when 
$k$ is not fixed. We can modify Algorithm~\ref{augmentation_algo} for RSNDP 
by solving this LP relaxation instead and following the same rounding algorithm. 
This, along with the reduction to Bulk-SNDP discussed in Section~\ref{sec:intro}, 
completes the proof of Corollary~\ref{cor:rgroup}.

\section{Conclusion and Open Problems}

In this paper we consider two network design problems in non-uniform
fault models: flexible graph connectivity and bulk-robust network
design.  We provide the first constant-factor approximation for some
special cases of flexible connectivity in the single pair and spanning
settings, and provide the first nontrivial approximation algorithms for
Bulk-Robust SNDP.  We also discuss natural LP relaxations for these
problems and obtain integrality gaps. Our work shows both the
potential and limitations of extending known algorithmic approaches
for connectivity problems to the non-uniform setting.

Non-uniform network design is a fairly new area of study with several
open problems. We highlight one important one. Is there a constant
factor approximation for fixed width/connectivity?  Achieving this
even for the special cases of global or single pair connectivity would
be a significant step forward. We are also interested in hardness of
approximation. Is it possible to rule out a constant factor
approximation for fixed width Bulk-Robust SNDP?  For large/unbounded
width, there are problems such as global connectivity FGC for which a
constant factor approximation is not ruled out. We are hopeful that
progress on these questions will also lead to the development of new
ideas and techniques. In fact,  there have already been some nice
developments in extending algorithmic ideas such as 
primal-dual to larger class of problems 
\cite{BansalCGI22,Nutov23}.

\mypara{Acknowledgements:} We thank Qingyun Chen for clarifications on
a proof in \cite{ChenLLZ22}. We thank Joseph Cheriyan and Ishan Bansal
for pointers and helpful comments on flexible graph connectivity. The
initial impetus for our work on this topic came from \cite{BoydCHI22}.
We thank Mik Zlatin for pointing out \cite{gupta2009online} that led
to some new results via the cut-covering lemma.

\bibliographystyle{plainurl}
\bibliography{fgc_paper}

\appendix
\section{$\Omega(k)$ integrality gap for $(1,k)$-Flex-ST}
\label{sec:integrality_gap_1k}
For $(1,k)$-Flex-ST in directed graphs, Adjiashvili et al.\ \cite{AdjiashviliHMS20} 
showed an integrality gap of $(k+1)$ for an LP relaxation similar to the
one we described in this paper. They also showed a
poly-logarithmic factor inapproximability via a reduction from
directed Steiner tree. It is natural to ask whether the undirected
version of $(1,k)$-Flex-ST is super-constant factor hard when $k$ is large. As an
indication of potential hardness, we show an $\Omega(k)$-factor
integrality gap for the LP relaxation for $(1,k)$-Flex-ST. This is via
a simple modification of an example from \cite{ChakCKK15} that showed
an $\Omega(R)$-factor integrality gap for single-pair capacitated
network design, where $R$ is the connectivity requirement.

Consider the following graph $G$ (shown in Figure \ref{1k_integrality_image}).
The vertex set $V$ is $\{s, t\} \cup \{v_i: i \in [k+1]\}$. We add two parallel 
unsafe edges from $s$ to each $v_i$ of cost $\frac 1 2$ each, and one 
safe edge from each $v_i$ to $t$ with costs $k+1$ each. Recall that the goal 
is to find a cheap subgraph of $G$ where $s$ and $t$ remain connected even after 
the removal of any set of $k$ edges.

\begin{figure}
  \begin{center}
    \includegraphics[width = 0.4\linewidth]{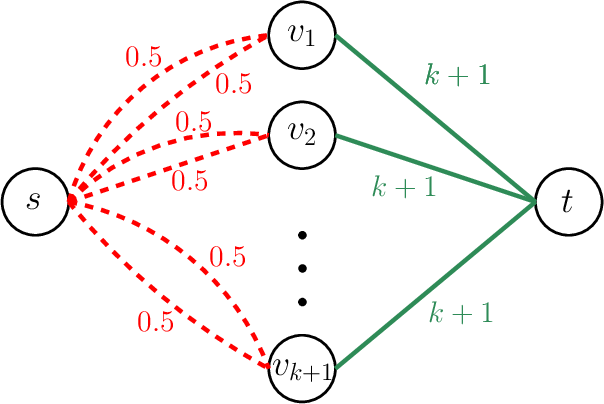}
  \end{center}
  \caption{Integrality gap for $(1,k)$-Flex-ST}
  \label{1k_integrality_image}
\end{figure}

\begin{claim}
Any optimal integral solution needs at least $\frac {k+1}2$ safe edges.
\end{claim}
\begin{proof}
Consider any integral solution $F \subseteq E$ where $|F \cap \calS| < \frac{k+1}2$. 
Let $I = \{i: \{v_i, t\} \notin F\}$ be the set of indices corresponding to $v_i$ 
not incident to a safe edge in $F$, and let $S = \{s\} \cup \{v_i: i \in I\}$.
Then, $\delta_F(S)$ is exactly the set of all unsafe edges from $s$ to $v_j$ for 
$j \notin I$. By construction, $|I| \geq k+1 - |F \cap \calS|$, so $|I| > \frac{k+1}2$.
Thus, $|\delta_F(S)| \leq 2 \cdot (k+1 - |I|) < k+1$. Since $\delta_F(S)$ has no 
safe edges, $S$ is a violated cut, so $F$ is not a feasible solution.
\end{proof}

\begin{claim}
Let $x \in [0,1]^{|E|} : x_e = 
\begin{cases}
    1 & e \in \calU \\
    \frac{2}{k+1} &e\in \calS
\end{cases}$. 
$x$ is a feasible LP solution.
\end{claim}
\begin{proof}
Let $S \subseteq V$ be any arbitrary cut such that $s \in S$, $t \notin S$. 
The capacitated constraint for $(1, k)$-flex connectivity is:
\[ (k+1) \sum_{e \in \delta_\calS(S)} x_e + \sum_{e \in \delta_\calU(S)} x_e \geq k+1.\]
Substituting the values of $x_e$, we get
\[(k+1) \sum_{e \in \delta_\calS(S)} x_e + \sum_{e \in \delta_\calU(S)} x_e = 
(k+1) \left(\frac 2 {k+1}\right)|\delta_\calS(S)| + |\delta_\calU(S)| = 
2|\delta_\calS(S)| + |\delta_\calU(S)|.\]
Notice that $\delta(S)$ has one safe edge for each $v_i \in S$, since all safe edges 
are from $v_i$ to $t$ and $t \notin S$, so $|\delta_\calS(S)| = |S| - 1$. 
Similarly, $\delta(S)$ has two unsafe edges for every $v_i \notin S$, since 
$s \in S$. Thus $|\delta_\calU(S)| = 2(k+1 - (|S| - 1)) = 2(k - |S| + 2)$. 
Thus $2|\delta_\calS(S)| + |\delta_\calU(S)| = 2(|S| - 1) + 2(k - |S| + 2) 
= 2(k+1) \geq k + 1$.

For the remaining set of constraints, let $B \subseteq \calU$, $|B| = k$. 
We want to show that $\sum_{e \in \delta(S) - B} x_e \geq 1$. We case on $|S|$. 
If $|S| \geq \frac{k+3}2$, then by the analysis above, $|\delta_\calS(S)| = 
|S| - 1 \geq \frac{k+1}2$, so $\sum_{e \in \delta_\calS(S)} x_e \geq 
\frac{2}{k+1} \cdot \frac{k+1}{2} = 1$. Since $B \subseteq \calU$, 
$\sum_{e \in \delta(S) - B} x_e \geq 1$. Suppose instead that $|S| < \frac{k+3}2$. 
Then, since $x_e = 1$ for all unsafe edges,
\begin{align*}
    \sum_{\delta_\calU(S) - B} x_e \geq |\delta_\calU(S)| - |B| 
    = 2(k - |S| + 2) - k = k + 4 - 2|S| > k+4 - (k+3) = 1.
\end{align*}
In either case, we get our desired result.
\end{proof}

Notice that the cost of any optimal integral solution is at least 
$(k+1) \frac {k+1}2 = \frac{(k+1)^2}2$. However, the optimal fractional solution 
has total cost $3(k+1)$. This shows an $\Omega(k)$ integrality gap.

\section{A $5$-approximation for $(2,2)$-Flex-ST}
\label{sec:st22}

In this section, we show that the algorithm we described for Flex-ST in 
Section \ref{subsec:st} (with minor modifications) gives a $5$-approximation 
for $(2,2)$-Flex-ST. 

Recall that in Section \ref{subsec:st}, we begin with an instance of 
$(p,(p+q))$-Cap-ST where every safe edge is given a capacity of $p+q$ and 
every unsafe edge is given a capacity of $p$. In this case, that corresponds to 
an instance of $8$-Cap-ST where every safe edge has capacity $4$ and every 
unsafe edge has capacity $2$. This is equivalent to an instance of $4$-Cap-ST 
where every safe edge has capacity $2$ and every unsafe edge has capacity $1$.
As mentioned in Section \ref{sec:intro}, since the minimum cost $s$-$t$ flow 
problem can be solved optimally, there exists a $\max_{e}u_e$ = 
$2$-approximation for this problem. Let $F \subseteq E$ be such a solution, 
and note that $\cost(F) \leq 2 \cdot \opt$. 

\begin{lemma}
\label{lem:st_22_starting_criteria}
For every set $A \subseteq V$ separating $s$ from $t$, at least one of the 
following is true: (1) $|\delta_{F \cap \calS}(A)| \geq 2$,
(2) $|\delta_{F}(A)| \geq 4$,
(3) $\delta_{F}(A)$ has exactly two unsafe and one safe edge.
\end{lemma}
\begin{proof}
This is easily verified by casing on the number of safe edges in 
$\delta_{F}(A)$. If |$\delta_{F \cap \calS}(A)| = 0$, since capacity of
$\delta_F(A)$ is $4$, it contains at least four unsafe edges. If 
$|\delta_{F \cap \calS}(A)| = 1$, then the safe edge provides capacity of 
$2$ and to reach capacity $4$, $\delta_F(A)$ contains at least two unsafe 
edges. Else |$\delta_{F \cap \calS}(A)| \ge 2$ as desired.
\end{proof}

Notice that case (3) in Lemma \ref{lem:st_22_starting_criteria} is the only 
case where $A$ is violated. In particular, this tells us that $F$ is a 
feasible solution to $(2,1)$-Flex-ST and that the only nontrivial iteration of 
Algorithm \ref{algo:augmentation_in_stages} is when $i = p-1 = 1$. 

Following Section \ref{subsec:st}, by symmetry we redefine $\calC$ to limit 
ourselves to the set of violated cuts containing $s$. We
consider a flow network on the graph $(V, F)$ with safe edges given a capacity 
of $2$ and unsafe edges given a capacity of $1$. Since $F$ satisfies the 
$4$-Cap-ST requirement, the minimum capacity $s$-$t$ cut and thus the 
maximum $s$-$t$ flow value is at least 4. By Lemma \ref{lem:st_22_starting_criteria}
all violated cuts have capacity 4, thus the maximum $s$-$t$ flow is exactly 4. 
By flow decomposition, we can decompose a maximum flow $f$ into a set 
$\calP = \{P_1, P_2, P_3, P_4\}$ of $4$ paths, 
each carrying a flow of 1, and we can find $\calP$ in polynomial time. 
For each $i \in [4]$, we define 
a subfamily of violated cuts $\calC^i$, where $A \in \calC^i$ iff 
the unique safe edge $e \in \delta_\calS(A)$ is contained in $P_i$.
By Lemma \ref{lemma:st_ringfamily}, each $\calC^i$ is a ring family. 
We will show that three of these subfamilies suffice to cover all violated 
cuts. 

\begin{lemma}
\label{lem:st_22_allcutscovered}
    $\calC \subseteq \calC_1 \cup \calC_2 \cup \calC_3$. 
\end{lemma}
\begin{proof}
    Let $A \in \calC$. By Lemma \ref{lem:st_22_starting_criteria}, $A$ has 
    exactly two unsafe edges and one safe edge. In particular, this means 
    that the capacity of $\delta(A)$ on the given flow network is 
    4, so all edges on $\delta(A)$ must be fully saturated for any max flow $f$. 
    This implies that the unique safe edge $e'$ in $\delta_\calS(A)$ must have
    flow value 2. Thus $e'$ belongs to at least two of the paths 
    $P_1, \dots, P_4$, so $e' \in P_1 \cup P_2 \cup P_3$. Therefore, $A$ 
    must be in $C_1 \cup C_2 \cup C_3$ as desired.
\end{proof}

\begin{theorem}
\label{thm:st_22_main}
    There exists a $5$-approximation for $(2,2)$-Flex-ST.
\end{theorem}
\begin{proof}
    We can solve the augmentation problems of finding the minimum cost
    subset $F_i \subseteq E \setminus F$ s.t. $\delta_{F_i}(A) \geq 1$ for
    all $A \in \calC^i$ for each $i = 1,2,3$. Since
    $\calC \subseteq \cup_{i=1}^3 \calC^i$ (by Lemma \ref{lem:st_22_allcutscovered}), 
    $F' = \cup_{i=1}^3 F_i$ is a
    feasible solution to the augmentation problem of finding the minimum
    cost subset of $E \setminus F$ s.t. $\delta_{F'}(A) \geq 1$
    for all $A \in \calC$. Therefore, $F \cup F'$ is a feasible solution
    to $(2,2)$-Flex-ST.
    From the discussion in Section \ref{sec:prelim}, we see that the three 
    corresponding augmentation problems can be solved exactly since each 
    $C^i$ is a ring family, so
    $\cost(F_i) \leq \opt$.  Thus,
    $\cost(F') \leq 3\opt$. Since $\cost(F) \leq 2 \opt$, 
    we obtain an overall cost of
    $\cost(F \cup F') \leq 5 \opt$.
\end{proof}

\end{document}